\documentclass[article,11pt]{els-PLF}
\usepackage{amsmath,amsfonts, amsthm, amssymb} 
        \newtheorem{theorem}{Theorem}[section]
        \newtheorem{proposition}[theorem]{Proposition}
        \newtheorem{lemma}[theorem]{Lemma}
        
        \newtheorem{definition}[theorem]{Definition}
        \newtheorem{remark}[theorem]{Remark}
\numberwithin{equation}{section}
\usepackage{a4wide}

\biboptions{square,comma}  
\journal{}

\usepackage{perpage} 
\MakePerPage{footnote}
\oddsidemargin 	.in
\evensidemargin .in
\topmargin    - 0.15in
\numberwithin{equation}{section}
  
\usepackage{dsfont, pxfonts, mathrsfs, verbatim}  
 
\usepackage{graphics}

\newcommand \auth {\textsc} 
\newcommand \jou  {\textit}

\newcommand \tr {\mathrm{tr }}
\newcommand \gthree {{}^{(3)}g}
\newcommand \uh {\widehat u}
\newcommand \mut {\widetilde \mu}

\newcommand \lbrac \llbracket 
\newcommand \rbrac \rrbracket 
\newcommand \R {\mathbb{R}}

\newcommand \RR {\mathbb{R}}
 
\newcommand \del  {\partial}  
\newcommand \loc {\text{loc}}
\newcommand \be {\begin{equation}}
\newcommand \ee {\end{equation}}

\newcommand \Rcal {\mathcal{R}}
\newcommand \Mcal {\mathcal{M}} 
 
\newcommand \Scal {\mathcal{S}}
\newcommand \Wcal {\mathcal{W}}
\newcommand \lam \lambda
\newcommand \sig \sigma
\newcommand \gam \gamma

\newcommand \Hcal {\mathcal{H}}

\newcommand \eps \epsilon
\newcommand \Lam \Lambda
 
\newcommand \la \langle
\newcommand \ra \rangle

\begin{document}
\title{Weakly regular Einstein--Euler spacetimes with Gowdy symmetry. 
The global areal foliation} 
\author{Nastasia Grubic\footnote{Laboratoire Jacques-Louis Lions
\& Centre National de la Recherche Scientifique, 
Universit\'e Pierre et Marie Curie (Paris 6),
4 Place Jussieu, 75252 Paris, France. 
{\it Email:} grubic@ann.jussieu.fr, contact@philippelefloch.org. 
{\it Blog: } philippelefloch.org. \hfill
\\
Accepted for publication in {\sl Archive for Rational Mechanics and Analysis.}} and Philippe G. LeFloch$^1$ 
}                   

\date{December 2012}

\begin{abstract}
 We consider weakly regular Gowdy--symmetric spacetimes on $T^3$ satisfying the Einstein--Euler equations of general relativity, and we solve the initial value problem when the initial data set has bounded variation, only, so that the corresponding spacetime may contain impulsive gravitational waves as well as shock waves. By analyzing, both, future expanding and future contracting spacetimes, we establish the existence of a global foliation by spacelike  hypersurfaces so that the time function coincides with the area of the surfaces of symmetry and asymptotically approaches infinity in the expanding case and zero in the contracting case. More precisely, the latter property in the contracting case holds provided the mass density does not exceed a certain threshold, which is a natural assumption since certain exceptional data with sufficiently large mass density are known to give rise to a Cauchy horizon, on which the area function attains a positive value. An earlier result by LeFloch and Rendall assumed a different class of weak regularity and did not determine the range of the area function in the contracting case. Our method of proof is based on a version of the random choice scheme adapted to the Einstein equations for the symmetry and regularity class under consideration. We also analyze the Einstein constraint equations under weak regularity. 
\end{abstract} 

\maketitle

\section{Introduction} 
  
We study here the class of Gowdy--symmetric spacetimes on $T^3$ when the matter model is chosen to be a compressible perfect fluid. We formulate the initial value problem when an initial data set is prescribed on a spacelike hypersurface and we search for a corresponding solution to the Einstein--Euler equations. In comparison with earlier works on vacuum spacetimes  \cite{EardleyMoncrief,Gowdy,IsenbergMoncrief,IsenbergWeaver,Ringstrom1,Ringstrom2,LeFlochSmulevici}
or on the coupling with the Vlasov equation of kinetic theory \cite{Andreasson,DafermosRendall,Rendall1,Rendall-crush,Smulevici}, dealing with compressible fluids is comparatively more challenging, since shock waves generically form in the fluid~\cite{RendallStahl} and one must work with weak solutions to the Einstein--Euler equations. 

In the present paper, building upon work by LeFloch and co-authors in \cite{LeFloch-prepa}--\cite{LeFlochStewart2}, 
we deal with the initial value problem associated with the Einstein--Euler equations under  weak regularity conditions, only: the first--order derivatives of the metric coefficients and the fluid variables are assumed to have bounded variation, which permits impulsive gravitational waves in the geometry and shock waves in the fluid. Analyzing, both, future expanding and future contracting spacetimes, we establish the existence of a global foliation by spacelike hypersurfaces, so that the time function coincides with the area of the surfaces of symmetry and asymptotically approaches $+\infty$ in the expanding case and $0$ in the contracting case. More precisely, the latter holds provided certain exceptional initial data are excluded which, instead, give rise to a Cauchy horizon \cite{Rendall-book,WainwrightEllis}. 
 
Recall that LeFloch and Rendall \cite{LeFlochRendall} established the existence of weak solutions to the initial-value problem, but 
assumed a different class of regularity and did not describe the interval of variation of the area function in the contracting case. 
The method of proof introduced in the present paper is based on a generalized version of the random choice scheme originally introduced by Glimm~\cite{Glimm} (cf.~the textbooks~\cite{Dafermos-book,LeFloch-book}) which we need to adapt to the Einstein equations and to the (Gowdy) symmetry class under consideration. In order to show the convergence of the proposed scheme, we derive a uniform bound on the sup-norm and the total variation of, both, first-order derivatives of the essential metric coefficients and the matter variables.  Metrics with low regularity were considered earlier in \cite{Christodoulou,GroahTemple} for the class of radially symmetric spacetimes, which, however, do not permit gravitational waves which are of main interest in the present work. 

We thus consider spacetimes $(\Mcal, g)$ satisfying the Einstein--Euler equations 
\be
\label{EE1}
{G_\alpha}^\beta = {T_\alpha}^\beta 
\ee
and study the initial value problem when an initial data set, describing the initial geometry and fluid variables, is prescribed on a spacelike hypersurface with $3$--torus topology $T^3$; moreover, we assume that the initial data set is Gowdy--symmetric, that is, are invariant under the Lie group $T^2$ with vanishing twist constants. (See Section~2, below, for the precise definitions.)
The left-hand side of \eqref{EE1} is the Einstein tensor $G_{\alpha \beta} := R_{\alpha\beta} - (R/2) g_{\alpha\beta}$ 
and describes the geometry of the spacetime, while $R_{\alpha\beta}$ denotes the Ricci curvature tensor associated with $g$. By convention, all Greek indices take here the values $0, \ldots, 3$.

The stress-energy tensor of a perfect compressible fluid reads 
\be
\label{stren}
{T_\alpha}^\beta := (\mu + p) \, u_\alpha u^\beta + p \, {g_\alpha}^\beta,
\ee
where $\mu$ represents the mass--energy density of the fluid and $u$ its future-oriented, unit time-like, velocity field. 
We also assume that the pressure $p$ depends linearly upon $\mu$, i.e.
\be 
\label{eq:pressure} 
p:=k^2 \mu,
\ee
in which the constant $k\in (0,1)$ represents the speed of sound in the fluid and does not exceed the speed of
light,  normalized to be unit.
  
Two main results are established in the present paper, for expanding spacetimes in Section~\ref{sec:6} (cf.~Theorem~\ref{maintheo})
and for contracting spacetimes in Section~\ref{sec:7}  (cf.~Theorem~\ref{maintheo2}). The main elements of proof are as follows: 

\begin{itemize}
  
\item Our first task is to formulate the initial value problem for the Einstein-Euler equations by solely assuming weak regularity conditions on the initial data set. Specifically, in areal coordinates $(t,\theta)$, the fluid variables as well as the first--order derivatives of the metric coefficients are assumed to have bounded variation in space, in the sense that their (first--order) distributional derivatives are bounded measures. This is so except for one metric coefficient which is denoted by $a$ (see Section~2, below) and describes the conformal metric of the quotient spacetime under the group action of the  Gowdy symmetry. The time-derivative $a_t$ has bounded variation in time, but no such regularity is imposed (nor available) for the spatial derivative $a_\theta$. Cf.~Sections \ref{sec:2} and \ref{sec:3}, below, for further details. 

\item Next, in Section~\ref{sec:4}, we study special solutions to the Einstein--Euler system. On one hand, we analyze the Riemann problem
(first studied in \cite{Thompson,SmollerTemple}, as far as the fluid variables in the flat Minkowski geometry are concerned) associated with a homogeneous version of the Einstein--Euler system. On the other hand, we study here an ordinary differential system which takes into account the source terms, only, after neglecting the propagation--related terms of the Einstein--Euler system. Here, a particular definition of the ``main geometric variables'' is essential in order for certain linear interaction estimates to hold and in order to avoid any resonance phenomena, due to the possibility of vanishing wave speeds. 

\item A generalization of the random choice scheme for the Einstein--Euler equations is introduced in Section~\ref{sec:5}, which allows us to establish a continuation criterion. In short, solutions to the Einstein--Euler equations are proven to exist as long as the mass-energy density variable $\mu$ (or, more precisely, the geometric coefficient $a$) does not blow-up. 

\item Based on the above arguments, we construct solutions to the initial value problem and investigate the global geometry associated with initial data sets generating future expanding (in Section~\ref{sec:6}) and future contracting (in Section~\ref{sec:7}) spacetimes. We prove that, in the expanding case, the areal foliation is defined for all values of the area function, while, in the contracting case, we establish the same global result for initial data with ``sufficiently small'' mass density. The latter is a natural assumption since certain exceptional data (with sufficiently large mass density) are known to give rise to a Cauchy horizon.  

\end{itemize} 

\section{Einstein--Euler spacetimes with weak regularity}  
\label{sec:2}

\subsection{Euler equations} 

Our first task is to introduce a definition of weak solutions to the Einstein--Euler equations. We focus here on the formulation 
of the Euler equations, and we refer to \cite{LeFlochSmulevici} for a detailed treatment of the (vacuum) Einstein equations.
We are given a smooth manifold $\Mcal$ (described by sufficiently regular transition maps) together 
with a locally Lipschitz continuous Lorentzian metric $g$. We suppose that this spacetime is foliated 
by compact spacelike hypersurfaces $\Hcal_t$ (diffeomorphic to a fixed three--manifold $\Hcal$), i.e.
$$
\Mcal = \bigcup_{t \in I} \Hcal_t,
$$
determined as level sets of  a locally Lipschitz continuous time function $t: \Mcal \to I \subset \RR$ ($I$ being an interval) so that $g^{\alpha\beta} \del_\alpha t$ is a future-oriented timelike vector. The lapse function $N$ and the future-directed, unit normal $n=(n^\alpha)$ are then defined as
$$
N := \big( -g^{\alpha\beta} \del_\alpha t \, \del_\beta t \big)^{-1/2}, 
\qquad
n^\alpha := -Ng^{\alpha\beta} \del_\beta t.
$$
These fields are measurable and bounded, only, that is, belong to $L^\infty_\loc(\Mcal)$. We assume, in addition,  
that all fields in $L^\infty_\loc(\Mcal)$ actually admit well-defined values on each hypersurface $\Hcal_t$, which also belong to 
$L^\infty(\Hcal_t)$. We write, in short, $N, n^\alpha \in L^\infty_\loc(L^\infty)$, by specifying first the regularity in time and, then, the regularity in space. In practice, these objects will also enjoy certain continuity properties in time, as we will state it below. 
Furthermore, we impose that the lapse function is also locally Lipschitz continuous ---which implies that the volume forms of, both, the Lorentzian metric and the induced Riemannian metric are also locally Lipschitz continuous.   

Under the above assumptions, we decompose the Lorentzian metric in the form 
$$
g = -N^2 dt^2 + \gthree, 
$$
in which $\gthree$ is the induced metric on the slices of the foliation. In local coordinates
 $(t,x^a)$ adapted to the foliation, we write $\gthree = (g_{ab})$ with $a,b = 1,2,3$. 
We introduce the Levi--Cevita connection $\nabla_\alpha$ associated with $g$,
 as well as the Levi--Cevita connection $ \nabla_a$ associated with $\gthree$, 
and we define the second fundamental form $k=k_t$ of a slice $\Hcal_t$ by  
$$
k(X,Y):=g(\nabla_X Y,n) = -g(\nabla_X n, Y)
$$ 
for any tangent vectors $X,Y$. 
Observe that, at this juncture, $\nabla_\alpha$, $\nabla_a$, and $k_{ab}$ belong to $L^\infty_\loc(L^\infty)$, only.

We now turn our attention to the Euler equations by first assuming that the fluid variables are locally Lipschitz continuous functions
(which is too strong an assumption to allow for shock waves and will be weakened in Definition~\ref{102}, below). Using the projection operator
$h^{\alpha\beta} := g^{\alpha\beta} + n^\alpha n^\beta$, we decompose the stress--energy tensor \eqref{stren}
by setting 
\be
\rho:=T^{\alpha\beta}n_\alpha n_\beta, 
\qquad 
j^\alpha := -T^{\gamma\beta}h_\gamma^\alpha n_\beta, 
\qquad 
S^{\alpha\beta}:=T^{\gamma\delta}h_\gamma^\alpha h_\delta^\beta, 
\ee
so that
$$
T^{\alpha\beta} = \rho n^\alpha n^\beta + j^\alpha n^\beta + j^\beta n^\alpha + S^{\alpha\beta}.
$$
In coordinates $x=(x^\alpha) = (t, x^a)$ adapted to the foliation, we have $n^0 = 1/N$, $n^a = 0$, $j^0 = 0$, and $S^{0\alpha}=0$ and, in order to simplify the notation, we define $\omega:=(\det(\gthree))^{1/2}$,  
so that $| \det(g) |^{1/2} = N \, \omega$. 

The Bianchi identities imply the Euler equations  $\nabla_\alpha T_\beta^\alpha=0$, which yield  
\be
\label{eq:22}
\int_\Mcal  T_\beta^\alpha \nabla_\alpha\xi^\beta \, dV_g = 0  
\ee
for every smooth and compactly supported vector field $\xi$.
Here, the volume form $dV_g = |\det (g)|^{1/2} \, dx = N \omega \, dx$ is locally Lipschitz continuous, 
while the derivative operator $\nabla$ is $L_\loc^\infty$, only.

On one hand, by taking the vector field $\xi$ to be normal, that is, $\xi = \varphi\del_t$ for some smooth and  
compactly supported function $\varphi$, we deduce from \eqref{eq:22} that  
$$
\aligned 
0=& \int_\Mcal T^\alpha_\beta \del_t^\beta \del_\alpha\varphi \,  dV_g + \int_\Mcal T^\alpha_\beta\nabla_\alpha \del_t^\beta \,\varphi\, dV_g
\\
=&\int_\Mcal N^2 \big( \rho n^\alpha + j^\alpha \big) \del_\alpha\varphi\, \omega dx 
- \int_\Mcal N \rho n_\beta n^\alpha\nabla_\alpha \del_t^\beta \,\varphi \, \omega dx 
 - \int_\Mcal N\, S^\alpha_\beta\nabla_\alpha \del_t^\beta \,\varphi \, \omega dx,
\endaligned
$$
thus  
\be
\label{50} 
\int_\Mcal \Bigg( 
N \, \rho \, \del_t\varphi  +  N^2 j^a\del_a\varphi + \big( \rho \, \del_t N+  N^2 \, S^{ab}k_{ab} \big) \,\varphi \Bigg) \, \omega dx = 0.  
\ee
On the other hand, by taking $\xi$ to be tangential to the slices, that is, having $\xi^0 = dt(\xi) = 0$, we find 
$$
\aligned 
0=&\int_\Mcal T^\alpha_\beta \,\nabla_\alpha\xi^\beta \,  dV_g 
\\
=&\int_\Mcal N \,\Big( \rho n^\alpha n_\beta + j^\alpha n_\beta + j_\beta n^\alpha + S^\alpha_\beta \Big) \, \nabla_\alpha\xi^\beta \, \omega dx  
\\
=&-\int_\Mcal  \rho \, \del_b N \,\xi^b dx 
+ \int_\Mcal  N \, k_{ab} j^a \xi^b\, \omega^2 dx 
+ \int_\Mcal N \, j_\beta n^\alpha\nabla_\alpha \xi^\beta\, \omega dx 
\\
& 
+ \int_\Mcal N \, S^a_b \,\nabla_a\xi^b \, \omega dx.
\endaligned 
$$  
The third term on the right--hand side reads
$$
\aligned
\int_\Mcal N\omega j_\beta n^\alpha\nabla_\alpha \xi^\beta\, dx 
= & \int_\Mcal  j_\beta \del_t\xi^\beta \, \omega dx 
         + \int_\Mcal N \ j_\beta \xi^\alpha\nabla_\alpha n^\beta \, \omega dx
\\
= &\int_\Mcal j_b \del_t\xi^b \, \omega dx - \int_\Mcal N \, k_{ab} j^a \xi^b \, \omega dx 
\endaligned
$$
and, therefore, 
\be
\label{51}
\int_\Mcal \Bigg( j_b \del_t\xi^b + N \, S^a_b \, \nabla_a\xi^b - \rho \, \del_b N \,\xi^b \Bigg) \, \omega dx = 0. 
\ee

At this juncture, we observe that \eqref{50}--\eqref{51} make sense under a weak regularity assumption on the fluid variables, which motivates us to introduce the following definition. 

\begin{definition} 
\label{102}
Consider a spacetime $(\Mcal, g)$ with locally Lipschitz continuous metric endowed with a foliation whose lapse function $N$ is also locally Lipschitz continuous. Let $\mu > 0$ be a scalar field and $u = (u^\alpha)$ be a future-oriented, unit timelike vector field, both of them having solely $L^\infty_\loc$ regularity. 
Then, $(\mu, u)$ is called a {\rm weak solution to the Euler equations} \eqref{stren} if the two identities 
\be
\label{Euler} 
\aligned
& \int_{t \in I} \int_{\Hcal}
\Big( N \, \rho(\mu, u) \, \del_t\varphi  +  N^2 \, j^a(\mu, u)\del_a\varphi + \big( \rho(\mu, u) \, \del_t N
  + N^2 \, S^{ab}(\mu, u)k_{ab} \big) \, \varphi \Big) \, dt dV_{\gthree} = 0, 
\\
& \int_{t \in I} \int_{\Hcal} 
\Big( j_b(\mu, u) \, \del_t\xi^b + N \, S^a_b(\mu, u) \, \nabla_a\xi^b - \rho(\mu, u) \, \del_b N \, \xi^b \Big) \, dt dV_{\gthree}  = 0,
\endaligned
\ee
hold for all smooth and compactly supported fields $\varphi$ and $\xi^a$. 
\end{definition} 

In other words, observing that $dtdV_{\gthree} = \omega \, dx$, the conditions in the definition mean that 
\be
\label{Euler1} 
\aligned
& \del_t \big( N \omega \, \rho \big) + \del_a \big( N^2 \omega \, j^a \big) 
   = \omega \, \Big( \rho \, \del_t N + N^2 \, k_{ab}\, \, S^{ab} \Big),
\\
& \del_t \big( \omega j_b \big) + \omega \nabla_a \big( N \, S^a_b \big) 
= - \omega \, \rho \, \del_b N 
\endaligned
\ee
hold in the distribution sense in the coordinates $(t, x^a)$. 
This latter statement does make sense since the product $\omega\nabla_a(N \, S^a_b)$ 
is understood geometrically in a weak sense as precisely stated in the above definition. 
Our assumption that the lapse function $N$ is locally Lipschitz continuous implies that the terms $\del_t N$ and $\del_b N$ are bounded on each compact time interval so that the right--hand sides of \eqref{Euler1} are well-defined as locally bounded functions. 


From now on, following \cite{LeFlochRendall}, we restrict attention to spacetimes with Gowdy symmetry on $T^3$.  
Introducing coordinates $(t,x^1,x^B)$ (with $B=2,3$) that are compatible with this symmetry so that 
$(x^2,x^3)\in T^2$ correspond to the $2$--surfaces generated by the Killing fields, we can now rewrite the Euler equations. 
As observed in \cite{LeFlochRendall}, the Einstein's constraint equations (stated below) imply the second and third components of the fluid velocity vanish, i.e. 
$$
u^2 = u^3 = 0. 
$$ 
Consequently, the remaining non-vanishing components of the stress--energy tensor read 
$$
\aligned
&\rho = N^2 (\mu + p) \, (u^0)^2 - p, 
\qquad 
&& j^1 = N\, (\mu + p) \, u^0u^1,
\\
&S^{11} = (\mu + p)(u^1)^2 + p \, g^{11}, 
\qquad 
&& S^{BC} = p \, g^{BC}. 
\endaligned
$$
The term $\omega \nabla_a( N \, S^a_b)$ arising in \eqref{Euler1} can be simplified, as follows: 
$$
\aligned
\omega \nabla_a( N \, S^a_1) & = \del_1( N \omega S^1_1) - N \omega \, \Gamma^b_{1a} \, S^a_b 
\\
&= \del_1( N \omega S^1_1) - N \omega \, \Big( \Gamma^1_{11} \, S^1_1 + \Gamma^C_{1B} \, S^B_C \Big), 
\endaligned
$$
where the Christoffel symbols $\Gamma^1_{11},\Gamma^C_{1B} \in L^\infty_\loc(L^\infty)$ are given by
$$
\Gamma^C_{1B} = \frac{1}{2} g^{CD} \, \del_1 g_{DB}, 
\qquad 
\Gamma^1_{11} = \frac{1}{2} g^{11} \, \del_1 g_{11}.
$$
We can finally write the Euler equations in the form
\be
\label{55}
\aligned
& \del_t \big( N \omega \, \rho \big) + \del_1 \big( N^2\omega j^1 \big) 
= N \omega \, \Sigma_0,
\\
& \del_t \big( \omega j_1 \big) + \del_1 \big( N \omega S^1_1 \big) 
= \omega \, \Sigma_1,
\endaligned
\ee
where 
\be
\label{rhse}
\aligned
\Sigma_0 := \, & N \, \big(k_{11} (\mu + p) (u^1)^2 + p \, \tr k \big) + \big( N (\mu + p) (u^0)^2 - p/N \big) \, \del_t N,
\\
\Sigma_1 := & - \big( N^2 (\mu + p) (u^0)^2 - p \big) \, \del_1 N  +  \frac{1}{2}N \, p \, g^{BC} \del_1g_{BC}              
\\
&  \hskip4.cm + \frac{1}{2} N \, \del_1 g_{11} \big( (\mu + p) (u^1)^2 + p \, g^{11} \big). 
\endaligned
\ee
According to Definition~\ref{102}, the equations \eqref{55}--\eqref{rhse} are required to hold in the distributional sense in the coordinates $(t,x^1)$. 


\subsection{Weak formulation of the Einstein equations in areal coordinates}

We express the spacetime metric in areal coordinates, that is, 
\be
\label{areal}
g = e^{2(\eta-U)}(-a^2d t^2 + d\theta^2)+e^{2U}(dx + Ady)^2+e^{-2U}t^2dy^2, 
\ee
which depends upon four scalar functions $U, A, \eta, a$ depending on the variables $t, \theta$, with $a>0$. We assume that
$\theta \in S^1 = [0, 1]$ (with a periodic boundary condition), and we distinguish between two classes of initial data imposed on a 
slice $\Hcal_{t_0}$ and we solve in the future direction $t \geq t_0$:  
\be
\aligned
& \text{Future expanding spacetimes: } t_0 >0, 
\\
& \text{Future contracting spacetimes: } t_0 <0.   
\endaligned
\ee
The global geometry of the spacetimes is markedly different in each case, but yet many fondamental estimates are similar (up to a change of sign). 
At this juncture, we may expect that the time variable should describe the interval $[t_0, +\infty)$ in the expanding case and the interval $[t_0, 0)$ in the contracting case, although we will actually see that dealing with the contracting case is more involved.

After a tedious but straightforward computation from the Einstein equations, one obtains the evolution equations 
\be
\label{evolu1} 
\aligned
\big( t \, a^{-1}U_t \big)_t - \big( t \, a U_\theta \big)_\theta 
& = \frac{e^{4U}}{2ta} \big( A^2_t - a^2A^2_\theta \big) + ta\Pi^U, 
\\
\big( t^{-1} \, a^{-1}A_t \big)_t - \big( t^{-1} \, aA_\theta \big)_\theta 
&= -4t^{-1}a^{-1} \big( U_tA_t - a^2U_\theta A_\theta \big) + a\Pi^A,
\\
\big( t \, a^{-1}\eta_t \big)_t - \big( t \, a(\eta + \log a)_\theta \big)_\theta 
& = 2t \, a U^2_\theta +\frac{e^{4U}}{2ta} A^2_t + t \, a \Pi^\eta(t),
\endaligned
\ee
understood in the sense of distributions, where the lower--order matter terms are given by  
$$
\aligned
\Pi^U:=&\frac{1}{2}e^{2(\eta-U)}(\rho - S_1 + S_2 - S_3), 
\\
 \Pi^A :=& 2e^{2(\eta-2U)} S_{23}, 
\qquad \Pi^\eta := e^{2(\eta-U)} \big( \rho - S_3 \big), 
\endaligned
$$  
in terms of the spatial part $S_{\alpha\beta}$ of the energy-momentum tensor $T_{\alpha\beta}$. Here, relevant components of $S_{\alpha\beta}$ are given  
with respect to the orthonormal frame 
$$
\aligned
& e^0 := a^{-1}e^{-(\eta-U)}\del_t,	
\quad 
&& e^1 = e^{-(\eta-U)}\del_\theta,	
\\ 
& e^2 := e^{-U}\del_x,	
\quad 
&& e^3 := \frac{e^{U}}{t}(-A\del_x + \del_y), 
\endaligned
$$ 
and read
\be 
\label{projected}
\aligned
& S_1 := S(e^1,e^1) =  e^{2(\eta-U)} \, (\mu + p)(u^1)^2+ p, \quad &&S_ {23}  := S(e^2,e^3) = 0,
\\
& S_2 := S(e^2,e^2) = p, \quad &&S_3  := S(e^3,e^3) = p. 
\endaligned
\ee
By introducing the {\rm scalar velocity} 
$v :=\frac{u^1}{au^0}$
and using $(au^0)^2-(u^1)^2=e^{-2(\eta-U)}$, we can express $u^0,u^1$ in terms of $v$: 
$$
e^{2(\eta-U)}(u^0)^2=\frac{1}{a^2(1-v^2)}, \quad
 \quad e^{2(\eta-U)}(u^1)^2=\frac{v^2}{1-v^2}.
$$
Therefore, the matter part of the evolution equations simplifies and we have 
$$
\Pi^U=\frac{1}{2}e^{2(\eta-U)}(\mu - p),
\qquad \Pi^A = 0,
\qquad \Pi^\eta = e^{2(\eta-U)}\frac{\mu + p}{1-v^2}.
$$
Finally, we can state Einstein's constraint equations 
\be
\label{constraint}
\aligned
&\frac{\eta_t}{t}= (U^2_t + a^2U^2_\theta) + \frac{e^{4U}}{4t^2}(A^2_t + a^2A^2_\theta)+ e^{2(\eta - U)}a^2\rho, 
\\
&\frac{\eta_\theta}{t}=-\frac{a_\theta}{ta} + 2U_tU_\theta + \frac{e^{4U}}{2t^2}A_tA_\theta - e^{3(\eta - U)}aj^1,
\\
&\frac{a_t}{t} = -a^3e^{2(\eta-U)}(\rho - S_1).
\endaligned
\ee
    
\begin{remark} The ``classical'' formulation of the Einstein equations in areal coordinates involves the operator 
$\del_{tt} - a^2 \del_{\theta\theta}$ applied to $U,A,$ and $\eta$. 
However, as pointed out in \cite{LeFlochRendall}, this formulation of the matter equations can not be used to handle 
{\sl weak solutions,} since the metric coefficients do not have enough regularity in the present context
and, for instance, the product $ a^2U_{\theta\theta}$ does not make sense.  
\end{remark} 


\subsection{Weak formulation of the Euler equations in areal coordinates}
 
For the derivation of the Euler equations in areal coordinates, we need the following expressions of the metric coefficients: 
$$
\aligned
& N=ae^{(\eta-U)}, 
\qquad 
&& g_{00}=-N^{2} = -a^2e^{2(\eta-U)},
\\
& g_{11}=e^{2(\eta-U)},
\qquad 
&&g^{11}=e^{-2(\eta-U)},
\\
& \det(^{(2)}g)=t^2,
\qquad 
&& \omega^2 = \det(^{(3)}g)=e^{2(\eta-U)}t^2,
\\
& \mathrm{tr } k=-{a^{-1}e^{-(\eta-U)}} \Big( \frac{1}{t} + \eta_t - U_t \Big), 
\quad  
&& k_{11}= -a^{-1}e^{\eta-U}(\eta_t - U_t). 
\endaligned
$$
We can then rewrite the right-hand sides \eqref{rhse} of the Euler equations and obtain
\begin{align*}
\Sigma_0'
:= \, & N\omega\,\Sigma_0 
=  \omega N^2(k_{11}(\mu+ p)(u^1)^2 + (\mathrm{tr}k)p) + \omega(N^2(\mu + p)(u^0)^2 - p)\del_t N
\\
	 = \, & -ate^{4(\eta-U)}(\eta_t - U_t)(\mu + p)(u^1)^2 - ate^{2(\eta-U)}(\frac{1}{t} + \eta_t - U_t)p
\\
	& +ate^{2(\eta-U)}(\frac{a_t}{a} + \eta_t - U_t)(a^2e^{2(\eta-U)}\big(\mu + p)(u^0)^2 - p\big)			
\\																										
	 = \, & a\,te^{2(\eta-U)}\Big(-\frac{1}{t}p + (\eta_t-U_t)(\mu - p) + \frac{a_t}{a}\big((\mu + p)\frac{1}{1-v^2} - p\big) \Big), 
\end{align*}
and
\begin{align*}
\Sigma_1'
:= \, & \omega \Sigma_1 = te^{\eta - U}\big( -(N^2(\mu + p)(u^0)^2 - p)\del_\theta N  +  \frac{1}{2}N p\,g^{BC}\del_\theta g_{BC}  
\\
& \hskip4.cm 
      + \frac{1}{2} N\del_\theta g_{11} \big((\mu + p)(u^1)^2 + p g^{11}\big) \big)
\\
	= \, & a\,te^{2(\eta - U)}\Big(-(\eta_\theta - U_\theta)(\mu - p) - \frac{a_\theta}{a}\big((\mu + p)\frac{1}{1-v^2} - p\big)\Big).
\end{align*}
Finally, we obtain
$$
\aligned
& \del_t\Big(a\,te^{2(\eta-U)}(\mu + (\mu + p)\frac{v^2}{1-v^2})\Big)+\del_\theta\Big(a^2\,te^{2(\eta-U)}(\mu + p)\frac{v}{1-v^2}\Big) = \Sigma_0',
\\
&
\del_t\Big(te^{2(\eta-U)}(\mu + p)\frac{v}{1-v^2}\Big) + \del_\theta\Big(a\,te^{2(\eta-U)}(p+(\mu + p)\frac{v^2}{1-v^2})\Big) = \Sigma_1'.
\endaligned
$$
Recalling the linear equation of state assumed for $p$, that is, $p=k^2\mu$,
 the weak form of Euler equations in areal coordinates takes the following final form 
\be
\label{60} 
 \del_t\Bigg(a\,te^{2(\eta-U)}\mu\frac{1+k^2v^2}{1-v^2}\Bigg)
  + \del_\theta\Bigg(a^2\,te^{2(\eta-U)}\mu\frac{(1+k^2)v}{1-v^2}\Bigg) 
= a\,te^{2(\eta-U)}\mu\Sigma''_0,
\ee
\be
\label{61}
\del_t\Bigg(te^{2(\eta-U)}\mu\frac{(1+k^2)v}{1-v^2}\Bigg) 
   + \del_\theta\Bigg(a\,te^{2(\eta-U)}\mu\frac{k^2 + v^2}{1-v^2}\Bigg) 
  = a\,te^{2(\eta-U)}\mu\Sigma''_1,
\ee
understood in the sense of distributions, 
where $\Sigma''_0,\Sigma''_1$ are given by
$$
\aligned
\Sigma''_0 & = -\frac{k^2}{t} + (\eta_t - U_t)(1-k^2) +\frac{a_t}{a}\frac{1+k^2v^2}{1-v^2},
\\
\Sigma''_1 & = -(\eta_\theta - U_\theta)(1-k^2) - \frac{a_\theta}{a} \frac{1+k^2v^2}{1-v^2}.
\endaligned
$$

\section{First-order formulation of the Einstein--Euler equations} 
\label{sec:3}

\subsection{Spacetimes with bounded variation}

For the analysis in the present paper, it is convenient to refomulate the Einstein equations as a first--order hyperbolic system supplemented with constraint equations. In fact, we will first drop the constraint equations and study the evolution part of this system. We emphasize that the Einstein constraint equations are used in the first-order evolution equations. 
 
We now arrive at a new formulation of the field equations as a system of first-order, hyperbolic balance laws supplemented with ordinary differential equations. To this end, we introduce the new variables
\be
\label{404} 
\aligned
& W_\pm := (U_t \mp aU_\theta),
\qquad  V_\pm := e^{2U} (A_t \mp aA_\theta), 
\endaligned
\ee
which we refer to as the {\sl first--order geometric variables,} 
to which we append the fluid variables $( \, \mut, v),$ and the geometric coefficients $(a, \nu)$. 
Here $\mut$ is obtained by scaling the energy density and we also introduce a new metric coefficient, as follows: 
$$
\nu = \eta + \log (a), \qquad \mut:=e^{2(\nu-U)}\mu.  
$$ 
Observe that the product containing $a_\theta$ on the right--hand side of the ``second'' Euler 
equation \eqref{61} may not be defined under the regularity assumptions we are interested in. However, if it is multiplied by $a$, 
the second Euler equation can be rewritten so that $a_\theta$ {\sl cancels out.} In this case, using our new notation, the source-terms $\Sigma'''_0, \Sigma'''_1$ of the Euler equations take the form  
\be
\label{52}
\aligned 
\Sigma'''_0 
&= \Sigma''_0 =  -\frac{k^2}{(1-k^2)t} - \frac{1}{2}(W_+ + W_-) + \frac{t}{2}(W_+^2  + W_-^2) + \frac{1}{8t}(V_+^2 + V_-^2),
\\
\Sigma'''_1 
&= a\Sigma'' + \frac{a_t}{a}\frac{(1+k^2)v}{1-v^2} + a_\theta \frac{k^2 + v^2}{1-v^2} 
\\
&
  = -a(\nu_\theta - U_\theta)(1 - k^2) + \frac{a_t}{a}\frac{(1+k^2)v}{1-v^2} 
\\ 
&= -\frac{1}{2}(W_+ - W_-) + \frac{t}{2}(W_+^2  - W_-^2) + \frac{1}{8t}(V_+^2 - V_-^2). 
\endaligned
\ee 

Without specifying yet (see below) the regularity of the spacetime metric and lapse function, we now collect, 
in the new notation, the first two evolution equations, the ordinary differential equation 
for $a$, and Euler equations. 

\begin{definition} 
\label{first-order}
The {\rm first-order formulation of the Einstein--Euler equations} for Gowdy--symmetric spacetimes 
consists of five evolution equations for the geometry 
\begin{align}
(a^{-1}tW_\pm)_t \pm (tW_\pm)_\theta =& \pm \frac{(W_+ - W_-)}{2a} + \frac{1}{2at}V_+V_- + \frac{t\mut}{2a}(1-k^2), 
\label{geometry1}
\\
(a^{-1}t^{-1}V_\pm)_t \pm ( t^{-1}V_\pm)_\theta =& \mp \frac{(V_+ - V_-)}{2at^2} - \frac{2}{ta} W_\pm V_\mp,
\label{geometry2}
\\
a_t =& -at\mut (1-k^2), 
\label{geometry3}
\end{align}
two equations for the fluid  
\be
\label{fluid}
\aligned
&\del_t\Bigg(a^{-1}t\mut\frac{1+k^2v^2}{1-v^2}\Bigg)+\del_\theta\Bigg(t\mut\frac{(1+ k^2)v}{1-v^2}\Bigg) 
= 
a^{-1}t\mut(1-k^2) \, \Sigma'''_0, 
\\
&\del_t\Bigg(a^{-1}t\mut\frac{(1+k^2)v}{1-v^2}\Bigg) + \del_\theta\Bigg(t\mut\frac{k^2 + v^2}{1-v^2}\Bigg) 
= 
a^{-1}t\mut(1-k^2) \, \Sigma'''_1, 
\endaligned
\ee
and three constraint equations: 
\be
\label{58} 
\aligned 
\nu_t & = \frac{t}{2}(W_+^2  + W_-^2) + \frac{1}{8t}(V_+^2 + V_-^2) + t\mut \, \frac{k^2 + v^2}{1-v^2},
\\
- a\nu_\theta & = \frac{t}{2}(W_+^2  - W_-^2) + \frac{1}{8t}(V_+^2 - V_-^2) + t\mut \, \frac{(1+k^2)v}{1-v^2},
\endaligned
\ee
and 
\be
\label{59}
\aligned
\big( ta^{-1}(\nu_t + t\mut (1 - k^2)) \big)_t -(ta\nu_\theta)_\theta =
&  \frac{t}{2a}(W_+ - W_-)^2 + \frac{1}{8at}(V_+ + V_-)^2 
\\
&  
+ t a^{-1} \, \mut \, \frac{(1+k^2)}{1-v^2}.
\endaligned
\ee  
\end{definition}

As we will show it, \eqref{58} are indeed constraint equations, that is, provided the initial data satisfy these equations, they then hold throughout the spacetime. This is a standard property; yet, in the present class of weak regularity, it is essential to check it, as we do here.  

We are now in a position to state the regularity of the spacetimes under consideration and to define our concept of weak solutions.  
We denote by $BV(S^1)$ the space of {\sl functions with bounded variation,}  that is, bounded functions $f: S^1 \to \RR$ whose derivative is a bounded measure.  The BV regularity for tensor fields is defined by considering each component in adapted coordinates. 

\begin{definition} 
A {\rm  foliated spacetime with bounded variation} (or BV foliated spacetime, in short)
 is a  foliated spacetime $(\Mcal, g)$ with locally Lipschitz continuous metric such that the areal metric coefficients 
$U_t, U_\theta, A_t, A_\theta, a, \nu_t, \nu_\theta$ belong to $L^\infty_\loc(BV(S^1)) \cap Lip_\loc(L^1(S^1))$. 
Such a spacetime is said to be an {\rm Einstein--Euler spacetime with bounded variation} 
if, in addition, there exist a scalar field $\mu$ and 
a unit vector field $u = (u^\alpha)$ in $L^\infty_\loc(BV(S^1)) \cap Lip_\loc(L^1(S^1))$ such that 
the first-order formulation of the Einstein--Euler equations (cf.~Definition~\ref{first-order}) holds in the sense of distributions. 
\end{definition}

Whenever necessary, we will tacitly extend functions defined on the circle $S^1=[0,1]$ to periodic functions defined on the real line $\RR$ with period $1$. 

\subsection{From first--order to second--order variables}
\label{subsec:ftos}

In this section, we show that it is actually sufficient to solve the set of essential equations \eqref{geometry1}--\eqref{fluid}, without imposing the constraints 
at the initial time.  
This is justified, as long as we can show that there exists a function 
$\nu$ satisfying the constraint equations \eqref{58}--\eqref{59} and we can recover the metric coefficients $U$ and $A$ from a given BV solution $(W_\pm, V_\pm, a, \mut, v)$ to the first--order system \eqref{geometry1} -- \eqref{fluid}. 

Indeed, by a direct calculation, it follows from Einstein-Euler equations that the right--hand side of the constraint equations \eqref{58}, that is, 
$$
\aligned
C_1 &:= \frac{t}{2}(W_+^2  + W_-^2) + \frac{1}{8t}(V_+^2 + V_-^2) + t\mut \, \frac{k^2 + v^2}{1-v^2},
\\
C_2 &:= \frac{t}{2}(W_+^2  - W_-^2) + \frac{1}{8t}(V_+^2 - V_-^2) + t\mut \, \frac{(1+k^2)v}{1-v^2},
\endaligned
$$
satisfies the compatibility condition
$$
(a^{-1}C_2)_t - (C_1)_\theta = 0. 
$$   
Therefore, the system  
$$
\aligned
\nu_t &= C_1,
\qquad \qquad 
a\nu_\theta = C_2,
\endaligned
$$
can be solved for $\nu$, by integrating the first (say) of the above equations 
\be
\nu(t, \theta) 
:= 
 \nu_0(\theta) + \int_0^t \Bigg( \frac{t}{2}(W_+^2  + W_-^2) + \frac{1}{8t}(V_+^2 + V_-^2) + t\mut\frac{k^2 + v^2}{1-v^2} \Bigg)(t', \theta) \, dt'. 
\ee 
Obviously, the distributional derivatives $\nu_t, \nu_\theta$ of $\nu$ belong to $BV(S^1)$, while $\nu$ is periodic, since $\nu_0$ 
is periodic and, by construction,   
$$
\del_t \int_{S^1} \nu_\theta \,d\theta = \int_{S^1} \del_t(a^{-1}C_2) d\theta = \int_{S^1} \del_\theta C_1 \, d\theta= 0.
$$ 
Moreover, by a straightforward calculation, it follows that \eqref{59} holds as well.  Observe that the above calculations only require multiplying the evolution equations for $W_\pm$ and $V_\pm$, which is allowed even for {\sl weak} solutions, since the 
principal part of these equations is a {\sl linear} operator and jump relations are unchanged. 

It remains to recover the metric coefficients $U, A$ from a BV solution in the variables $W_\pm, V_\pm, a$. Indeed, by subtracting 
the evolution equations for $W_\pm$, we obtain  
$$
\Big( a^{-1} ( - W_+ + W_-)\Big)_t + \Big(  ( W_+ + W_-) \Big)_\theta = 0 
$$
and, therefore, the system
$$
U_t = \frac{1}{2} ( W_+ + W_-), \qquad U_\theta = \frac{1}{2a} ( - W_+ + W_-), 
$$
can be solved in $U$, by setting
\be
\label{eq:409}
U(t, \theta) = U_0(\theta) + \frac{1}{2} \int_0^t ( W_+ + W_-)(t', \theta) \, dt'.
\ee
Clearly, first--order derivatives of $U$ are functions of bounded variation in $\theta$, and $U$ is periodic. The same properties hold for the coefficient $A$, and our argument is completed. 

\section{Special solutions to the Einstein--Euler equations} 
\label{sec:4}
\subsection{Homogeneous solutions}
\label{sec:41} 

In this section, we neglect the source terms in the equations \eqref{geometry1}--\eqref{fluid} and we temporarily assume 
$t=1$, that is $tW_\pm \rightarrow W_\pm$, $t^{-1}V_\pm \rightarrow V_\pm$, $t\mut \rightarrow \mut$, to facilitate the exposition. The results for the original variables are easily recovered after scaling by a suitable power of $t$. Hence, we study the system 
\be
\label{homog}
\aligned 
(a^{-1}W_\pm)_t \pm (W_\pm)_\theta &= 0,
\\
(a^{-1}V_\pm)_t \pm (V_\pm)_\theta &= 0,
\\
a_t &= 0, 
\\
\del_t\Bigg(a^{-1}\mut\frac{1+k^2v^2}{1-v^2}\Bigg)+\del_\theta\Bigg( \, \mut\frac{(1+k^2)v}{1-v^2}\Bigg) & = 0, 
\\
\del_t\Bigg(a^{-1}\mut\frac{(1+k^2)v}{1-v^2}\Bigg) +\del_\theta\Bigg( \, \mut\frac{k^2 + v^2}{1-v^2}\Bigg) & = 0. 
\endaligned
\ee 
The geometry variables  $(V_{\pm},W_\pm, a)$ satisfy linear hyperbolic equations, associated with the speeds $-a, 0, +a$ only,
whereas the last two equations for $( \, \mut, v)$ correspond to the system of special relativistic fluid dynamics. 
The fluid equations are coupled to the geometry only
through $a$ and, thus, can be considered separately from $(V_\pm, W_\pm)$.  

We are interested first in solving the {\sl Riemann problem}
for the above system expressed in terms of $u:=(W_\pm, V_\pm, a, \mut, v)$, that is, in solving the initial value problem 
with piecewise constant initial data
\be
\label{ivriemann}
u(t_0, \theta) = \begin{cases}
u_l, \quad & \theta < \theta_0, 
\\
u_r, \quad & \theta >\theta_0. 
\end{cases} 
\ee
Here, the constant vectors $u_l, u_r$ are prescribed within the physical region 
\be
\label{phys}
a>0,  \quad \mut>0,  \quad |v|<1, 
\ee
and an arbitrary point $(t_0, \theta_0)$ has been fixed. 

To solve the Riemann problem for the subsystem $(a, \mut, v)$, we will make use of the Riemann invariants $y, z, w$  
$$
\aligned
& y:= \log a, \qquad\quad
&& z := \frac{1}{2}\log \Bigg(\frac{1+v}{1-v} \Bigg) + \frac{k}{1+k^2}\log \,\mut, 
\\ 
& w:=\frac{1}{2}\log \Bigg(\frac{1+v}{1-v} \Bigg) - \frac{k}{1+k^2}\log \,\mut,
\endaligned
$$
This system is not strictly hyperbolic at hypersurfaces $v=\pm k$, 
where, in each case, two of the corresponding characteristic speeds   
$$
\lambda_0 = 0, \quad \qquad \lambda_1 = a\frac{v - k}{1 - v \, k}, \quad \qquad \lambda_2 = a\frac{v + k}{1 + v \, k},
$$   
coincide. On the other hand, both $\lambda_1$ and $\lambda_2$ are genuinely nonlinear and $\lambda_0$ is obviously linearly degenerate. 
Moreover, we have $\lambda_1<\lambda_2$ where both are strictly monotonically increasing with 
respect to $v$.  Up to an additional vector $a$, they correspond to the characteristic speeds of the special relativistic Euler equations, which we denote by 
$$
\widehat{\lambda}_{1,2}:=\frac{v \mp k}{1 \mp vk}.  
$$

\begin{proposition}[Riemann problem for the Einstein--Euler equations]
\label{homogeq}
The system \eqref{homog} admits a unique self-similar weak solution 
to the initial value problem \eqref{ivriemann}. 
All the variables contain a stationary discontinuity. In addition, $W_+, V_+$ contain jumps associated with 
the speed $a_r$, 
while $W_-,V_-$ contain jumps associated with the speed $-a_l$. The fluid variables $( \, \mut, v)$
have up to three additional waves comprising rarefaction and/or shock waves satisfying Lax shock condition. Moreover, 
the regions 
$$
A_M:=  \big\{(w,z) \, / \, -M\leq w\leq z\leq M \big\}, \quad M>0, 
$$
are invariant for the solutions of the Riemann problem associated with the fluid equations. 
In other words, if the Riemann data belong to a set $A_M$, 
then so does the corresponding solution of the Riemann problem.  
\end{proposition}

The proof of this proposition is given shortly below. 
We will mesure the strength $|\eps(u_l, u_r)|$ of the waves present in the solution of the Riemann problem 
$\Rcal(u_l, u_r)$ associated with Riemann data $u_l, u_r$ in terms of the strength vector $\eps(u_l,u_r)\in \R^7$. 
Its first five components contain the jumps $(a, W_\pm, V_\pm)$, whereas the last two contain the jumps of 
$\log \mut$ across shock/rarefaction waves of the fluid equations. 
The total strength is defined as the sum of the components of $\eps$, that is
$$
|\eps(u_l, u_r)|:= \sum_i|\eps_i(u_l, u_r)|.
$$
Recall that $W_\pm, V_\pm$ and $\mut$ are all {\sl constant across a stationary discontinuity} 
and that a possible jump in $\mut$ occurring when one of the rarefaction waves
crosses the plane of non--strict hyperbolicity would not contribute if included in the definition of $\eps$, since 
$\mut$ is strictly monotone along both rarefaction curves.  
For the geometry variables, the following statement is just the triangle inequality, 
whereas for $\mut$ the claim follows from the corresponding estimate for special relativistic fluids~\cite{SmollerTemple}. 

\begin{proposition}[Interaction estimates]
\label{inter}
For any three Riemann problems $\Rcal(u_l, u_m)$, $\Rcal(u_m, u_r)$, $\Rcal(u_l,u_r)$ associated with 
constant states $u_l$, $u_m$, $u_r$, one has 
$$
|\eps(u_l,u_r)|\leq |\eps(u_l,u_m)| + |\eps(u_m,u_r)|.
$$ 
\end{proposition}

\begin{lemma}
\label{odeinter}
The strength vector $\eps : \Omega \times \Omega \rightarrow \R^7$ is a smooth function of its arguments, and for all 
$u_l,u_r, u_l', u_r'$ in a compact subset of $\Omega$ one has the uniform estimate
$$
\aligned
\eps(u_l',u_r') = \eps(u_l,u_r) 
& + O(1)|\eps(u_l,u_r)|(|u_l' - u_l| + |u_r' - u_r|) 
\\
& + O(1) |(u_r'-u_l')-(u_r-u_l)|.
\endaligned
$$
\end{lemma}

\noindent{\sl Proof of Proposition~\ref{homogeq}.} 
We first study the fluid part of the above equations, that is we study the subsystem corresponding to $u :=(a, \mut, v)$. 
Since $\lambda_0$ is linearly degenerate, the 0-wave curve starting at $u_l$ can be obtained from Rankine-Hugoniot relations or directly 
as an integral curve to the eigenvector $R_0 := (1,0,0)$. Immediately, we have
$$
\mut = \mut_l, \qquad v = v_l.
$$
The $1$-- and the $2$--wave shock curves $\Scal_{1,2}$ can be obtained by a direct calculation from the Rankine-Hugoniot relations, following 
the exposition in \cite{SmollerTemple}. Then, it follows immediately that $a$ remains constant along both shock curves, but also that $\mu$ and $v$ are given by the corresponding special relativistic shock curves; hence, all of the standard results still hold true.  
For completness, we state $\Scal_1$ respectively $\Scal_2$ starting at $u_l$ in terms of the Riemann invariants $y,w, z$. 
We have 
$$
\aligned
y(u)-y(u_l) & = 0,
\\
w(u)-w(u_l) &= -\frac{1}{2}\log f_+(2K\beta) - \sqrt{\frac{K}{2}}\log f_+(\beta),
\\
z(u)-z(u_l) & = -\frac{1}{2}\log f_+(2K\beta) + \sqrt{\frac{K}{2}}\log f_+(\beta),
\endaligned
$$
and
$$
\aligned
y(u)-y(u_l) & = 0,
\\
w(u)-w(u_l) &= -\frac{1}{2}\log f_+(2K\beta) - \sqrt{\frac{K}{2}}\log f_-(\beta),
\\
z(u)-z(u_l) &= -\frac{1}{2}\log f_+(2K\beta) + \sqrt{\frac{K}{2}}\log f_-(\beta),
\endaligned
$$
where $f_\pm(\beta) := 1 + \beta\big(1\pm\sqrt{1+2/\beta} \,\big)$ is defined for all $\beta \geq 0$,   
while $K:=2k^2/(1+k^2)^2$ is a constant.

For the shock speeds $s_i$, it can be checked that  
$s_i(a,v; a_l,v_l) = a \, \widehat{s}_i(v;v_l)$, 
where $\widehat{s}_i(v;v_l)$ have the same form as the shock speeds of the special relativistic equations starting at $v_l$. 
Therefore, Lax shock conditions
$$
\lambda_i(a,v) \leq s_i(a,v; a_l, v_l) \leq \lambda_i(a_l, v_l),
$$  
immediately follow from the corresponding standard result. On the other hand, for the rarefaction curves $\Rcal_1, \Rcal_2$ starting at $u_l$ we have
$$
\Rcal_1 = \big\{u \, / \, y(u) = y(u_l), \, z(u) = z(u_l), \, w(u)\geq w(u_l) \big\},
$$
$$
\Rcal_2 = \big\{ u \, / \, y(u) = y(u_l), \, w(u) = w(u_l), \, z(u)\geq z(u_l) \big\},
$$
which coincide with the rarefaction curves of the special relativistic equations. Finally, the composite wave curves, defined as $\Wcal_i\equiv \Scal_i\cup \Rcal_i$, are of class $C^2$ and convex; the parametrized curve 
$\Wcal_1$ is monotonically decreasing, while $\Wcal_2$ is monotonically increasing. From now on, we denote the $(w,z)$--components of 
$\Wcal_i$ by $\widehat{\Wcal}_i$.

The form of the curves $\widehat{\Wcal}_i$ suggests that they can be regarded as wave curves of the special relativistic fluid dynamics; 
hence, for given Riemann data $u_l, u_r$, there exists a unique point 
$$
( \, \mut_m, v_m)\in \widehat{\Wcal}_1( \, \mut_l, v_l)\cap \widehat{\Wcal}_2( \, \mut_r, v_r), 
$$ 
such that the special relativistic speed of the first wave at $v_m$ remains strictly smaller 
than the special relativistic speed of the second one at $v_r$. 
Therefore, depending on their respective signs, it only remains to decide where to put the stationary discontinuity in order to retain the correct wave speed relations. In terms of $\Wcal_i$, this amounts to jumping from a point $( \, \mut,v)$ on the plane $a=a_l$ to the same point on the plane $a = a_r$, whereby all the relevant speeds $\widehat{s}_i, \widehat{\lambda}_i$ prior to the jump get multiplied by $a_l$ and all following 
it by $a_r$. 

For instance, if we have $( \, \mut_m, v_m)\in \widehat{\Scal}_1( \, \mut_l, v_l)$ followed by $( \, \mut_r, v_r)\in \widehat{\Scal}_2( \, \mut_m, v_m)$
and $0<\widehat{s}_1(v_m)<\widehat{s}_2(v_r)$ we have to put the stationary discontinuity afirst, i.e., we have four constant states
$$
(a_l, \mut_l, v_l), \quad (a_r, \mut_l, v_l), \quad (a_r, \mut_m, v_m), \quad (a_r, \mut_r, v_r).
$$ 
Observe that, in the limiting case $\widehat{s}_1(v_m) = 0$, the middle state $(a_r, \mut_l, v_l)$ cancels out and the solution consists of the three constant states  
$$
(a_l, \mut_l, v_l),\quad (a_r, \mut_m, v_m),\quad (a_r, \mut_r, v_r).
$$
We obtain all the other cases similarly, except if either $\widehat{\Rcal}_1$ crosses the
plane $v = k$ or $\widehat{\Rcal}_2$ crosses $v = -k$, i.e.~if one of the characteristic speeds changes sign.  
Observe that it is not possible for both cases  to occur simultaneously, since $v\geq v_l$ in both cases. 
If one of them occurs however, the solution consists of five constant states, where the stationary discontinuity has to be put at 
$v = \pm k$. If we take, for example, the case of two rarefaction waves, where $\widehat{\Rcal}_1$ crosses the line $v=k$ 
at some $\mut_k$, then we have
$$
(a_l, \mut_l, v_l), \quad (a_l, \mut_k, k), \quad (a_r, \mut_k, k), \quad (a_r, \mut_m, v_m), \quad (a_r, \mut_r, v_r).
$$

We now turn to the geometric variables. The component $W_+$ of the Riemann problem consists of 
at most three constant states where a stationary discontinuity, across which $W_+$ is constant, 
is followed by a jump discontinuity of speed $a^r$ across which $a$ is constant. 
 We handle $W_-$ similarly and only note that here a jump discontinuity with speed $-a^l$ is 
followed by a stationary discontinuity. An analogous conclusion holds for $V^\pm$ and the proof is completed.   

\begin{proof}[Proof of Lemma~\ref{odeinter}]
Since $\eps(u_l,u_r) = 0$ if $u_l = u_r$ we have
$$\eps(u_l,u_r) = \int_0^1\frac{\partial \eps}{\partial u_r}(u_l, (1-\tau)u_l + \tau u_r)(u_r - u_l)d\tau$$
and analogously for $\eps(u_l', u_r')$; hence, we find 
\be
\aligned
&\eps_i(u_l',u_r')- \eps_i(u_l,u_r) 
\\
& = 
\int_0^1
\Bigg(\frac{\partial \eps}{\partial u_r}(u_l', (1-\tau)u_l' + \tau u_r')-\frac{\partial \eps}{\partial u_r}(u_l, (1-\tau)u_l + \tau u_r)\Bigg)
\, (u_r - u_l) \, d\tau 
\\
& \quad + \int_0^1\frac{\partial \eps}{\partial u_r}(u_l', (1-\tau)u_l' + \tau u_r') \, \big((u_r' - u_l') - (u_r - u_l)\big)d\tau
\\
& = O(1)\, \big( |u_l' - u_l| + |u_r' - u_r|\big) \, |u_r - u_l| + O(1) \, \big|(u_r' - u_l')- (u_r - u_l) \big|. 
\endaligned
\ee
On the other hand, we have 
$|\eps(u_l,u_r)| = O(1)|u_r - u_l|$, 
and the claim thus follows.
\end{proof}

\subsection{Spatially independent solutions}  

We set $u = (W_\pm, V_\pm, a, \mut, v)$ and study the following system of ordinary differential equations (ODE, in short) 
\be
\label{57} 
u'(t) = g(u(t), t), \qquad u(t_0) = u_0,
\ee
where $g$ represents the right--hand side of the equations (\ref{geometry1})-(\ref{fluid}) and we assume $a_0, \mut_0>0$ 
and $|v_0|<1$. More precisely, for the geometry variables we have
\be
\label{odegeometry} 
\aligned
(a^{-1}tW_\pm)' &= \pm \frac{W_+ - W_-}{2a} + \frac{1}{2at}V_+V_- + \frac{t\mut }{2a}(1-k^2), 
\\
(a^{-1}t^{-1}V_\pm)' &=  \, \mp \frac{V_+ - V_-}{2at^2} - \frac{2}{at} W_\pm V_\mp,  
\endaligned
\ee
while, for the fluid variables, 
\be
\label{odefluid}
\aligned
\Bigg(a^{-1}t\mut\frac{1+k^2v^2}{1-v^2}\Bigg)' 
= a^{-1} t\mut(1-k^2) \Big(& -\frac{k^2}{(1-k^2)t} - \frac{1}{2}(W_+ + W_-) 
\\
& + \frac{t}{2}(W_+^2  + W_-^2) + \frac{1}{8t}(V_+^2 + V_-^2) \Big), 
\\
\Bigg(a^{-1}t\mut\frac{(1+k^2)v}{1-v^2}\Bigg)' 
 = a^{-1} t \mut(1-k^2) \Big( &-\frac{1}{2}(W_+ - W_-) 
\\
& + \frac{t}{2}(W_+^2  - W_-^2) + \frac{1}{8t}(V_+^2 - V_-^2) \Big). 
\endaligned
\ee 
Finally, the component $a$ satisfies 
$$
a' = -at\mut (1-k^2).
$$

Observe that the above system admits an {\sl energy functional}
$$
E := \frac{1}{2a}(W^2_+ + W^2_-) + \frac{1}{8at^2}(V^2_+ + V^2_-) + \frac{\mut}{a}\frac{1+k^2v^2}{1-v^2},
$$
whose time-derivative reads 
$$
E' 
= - \frac{1}{t}\Bigg( \frac{1}{2a}(W_+ + W_-)^2 + \frac{1}{8at^2}(V_+ - V_-)^2 + \frac{\mut}{a}\frac{(1+k^2)}{1-v^2} \Bigg).
$$
Hence, $E$ is non-increasing in (areal) time in the expanding case, while it is 
non-decreasing in the contracting case. We now establish (for the simplified model under consideration in the present section, at least) that, 
as long as the coefficient $a$ remains bounded, the other geometric and fluid variables can not blow--up. 

\begin{lemma}[No blow-up in finite time. Preliminary version] 
\label{406} 
Fix a time $t_0$ satisfying $t_0>0$ in the expanding case and $t_0 < 0$ in the contracting case. 
Fix also some initial data $u_0$ satisfying the physical constraints
\be
\label{phys1}
0<a_0,  \qquad 0<\mut_0, \qquad  |v_0|<1,
\ee
and consider the corresponding solution of the ordinary differential system \eqref{57} defined on 
its maximal interval of existence $[t_0,T)$. (Here, $0 < t_0 < T$ in the expanding case
and $t_0 < T < 0$ in the contracting case.) Then, the following results hold: 
\begin{itemize}
\item In the expanding case, the solution is defined for all time up to $T=+\infty$ 
and 
\be
\label{301}
\lim_{t \to +\infty} W_\pm(t) = \lim_{t \to +\infty} {V_\pm(t) \over t} = 0,
\ee
\be
\label{302}
{1 \over a_0} \leq  {1 \over a(t)},
\ee
and 
\be
\label{303}
-1 < v < 1, \qquad \qquad 0 < \mu < + \infty. 
\ee
\item In the contracting case, $T$ may be finite; however, as long as $a(t)$ remains bounded, 
the variables $W_\pm, V_\pm$ remain bounded and $\mut, v$ remain strictly inside the physical domain; 
in other words, if $T \in (t_0, 0)$ and $a$ is defined on the interval $[t_0, T)$ and remains bounded on that interval, i.e.~{\rm provided}
\be
\label{306}
a_0 \leq \sup_{t \in [t_0, T)} a(t) < + \infty,
\ee
then  the solution $u=u(t)$ extends by continuity to the closed interval $[t_0, T]$
and, therefore, \eqref{301} and \eqref{303} hold in the interval $[t_0, T]$.  
\end{itemize}
\end{lemma}

In particular, this proposition shows that $(\, \mut, v)$ remain in the physically admissible region, i.e.
$$
0<\mut, \qquad |v|<1.
$$ 

\begin{proof} We will establish here a slightly stronger result than the one stated in the proposition and will be able to control 
the time-dependence of our upper bounds, which will turn out to be useful later on. 
We first rewrite the two fluid equations in terms of $\log \mut$ and $\log \frac{1+v}{1-v}$: 
$$
\aligned
\Big (\log ( t \, \mut) \Big)' & = \frac{(1-k^2)}{1-k^2v^2} \, f(W_\pm, V_\pm, v, t) - (1-k^2) \, t \, \mut,
\\
\Bigg( \log \frac{1+v}{1-v} \Bigg)' &= \frac{2(1-k^2)}{(1-k^2v^2)} \, g(W_\pm, V_\pm, v, t),
\endaligned
$$
with
$$
\aligned
f(W_\pm, V_\pm, v, t) 
& = -\frac{k^2}{t}\frac{(1+v^2)}{(1-k^2)} -2v\big( \,-\frac{1}{2}(W_+ - W_-) + \frac{t}{2}(W_+^2  - W_-^2) + \frac{1}{8t}(V_+^2 - V_-^2)\big)
\\
&\quad+(1+v^2)\big( - \frac{1}{2}(W_+ + W_-) + \frac{t}{2}(W_+^2  + W_-^2) + \frac{1}{8t}(V_+^2 + V_-^2)\big), 
\endaligned
$$
and 
$$
\aligned
g(W_\pm, V_\pm, v, t) 
& = \frac{k^2}{t}\frac{v}{(1-k^2)} - v\big(\, -\frac{1}{2}(W_+ + W_-) + \frac{t}{2}(W_+^2  + W_-^2) + \frac{1}{8t}(V_+^2 + V_-^2)\big) 
\\
&\quad + \frac{1+k^2v^2}{1+k^2}\big(\,-\frac{1}{2}(W_+ - W_-) + \frac{t}{2}(W_+^2  - W_-^2) + \frac{1}{8t}(V_+^2 - V_-^2)\big).
\endaligned
$$

For the expanding case $t_0>0$, the function $a$ is non-increasing and thus bounded by the initial value. 
Assume that on some time interval $[t_0, T_1)$ we have 
$0<\mut$ and $|v|<1$. 
Then, $E$ is decreasing and we have that $a^{-1}W^2_\pm$ and $a^{-1}t^{-2}V^2_\pm$ are bounded and therefore 
also $W_\pm, V_\pm$. Consequently, 
$\log \mut$ and $\log \frac{1+v}{1-v}$ remain bounded as well, since all the relevant geometry expressions 
are bounded. Therefore, there are constants $C, M>0$ depending on $T_1$ and the supremum bounds on the geometry such that 
$$ 
0< \frac{1}{C} \leq \mut \leq C, \qquad |v|\leq M<1.
$$    
Finally, the differential equation for $a$ implies that $a$ is decreasing and remains bounded away from zero 
$$
0<K\leq a(t)\leq a_0 
$$ 
for some constant $K>0$ and for all $t<T_1$. 

Next, we turn our attention to the contracting case $t_0<0$ and we assume $\sup _{t\in[t_0, T]}|a|<+\infty$ for 
some $T<0$. Suppose, moreover, that on the interval $[t_0, T)$ one has  
$0<\mut$ and $|v|<1$. The energy is now increasing and we have the inequality 
$$
E(t)'\leq -\frac{2}{t}E(t),
$$ 
which, by Gronwall lemma, implies 
$$
E(t)\leq E(t_0)\frac{t_0^2}{t^2}.
$$
Hence, we have 
$$
W_\pm^2(t) \leq a(t) 2E(t_0)\frac{t_0^2}{t^2}, \quad V_\pm^2(t)\leq 8E(t_0)t_0^2a(t),
$$
that is, the geometry variables remain bounded. On the other hand, this implies as in the expanding case that $v$ is 
bounded away from $\pm 1$, and $\mut$ from $0$ and $+\infty$. 
\end{proof}

The bounds in Lemma~\ref{406} can be stated in a quantitative form, which will be useful when dealing with approximate solutions later. 
We introduce the energy 
$$
h_1 := \frac{(tW_+ - 1/2 )^2}{2at^2} + \frac{(tW_- - 1/2 )^2}{2at^2} + \frac{V_+^2}{8at^2} + \frac{V_-^2}{8at^2}
$$
and the energy flux 
$$
g_1:= \frac{(tW_+ -1/2)^2}{2at^2}  - \frac{(tW_- -1/2)^2}{2at^2}  + \frac{V_+^2}{8at^2} - \frac{V_-^2}{8at^2},
$$
and observe that they satisfy (for spatially independent solutions) 
\be
\label{eq:hc} 
\aligned
h_1'
&= \frac{a_t}{a}h_1 - \frac{2k_1}{t}, 
\\g_1' 
 &= \frac{a_t}{a}g_1 - \frac{g_1}{t}, 
\endaligned
\ee
in which
$$
k_1 := \frac{(t(W_+ + W_-) - 1)^2}{4at^2} + \frac{(V_+ + V_-)^2}{16at^2}
$$
is non--negative, with $0 \leq k_1 \leq h_1$. 
We then introduce the variables\footnote{These variables are natural also for the homogeneous Riemann problem studied in Section~\ref{sec:41} and are easily checked to satisfy the invariant domain principle,
like the Riemann invariants $w,z$ for the fluid; cf.~Proposition~\ref{homogeq}.} 
\be
\label{396} 
H^\pm := a \, ( h_1 \pm g_1 ). 
\ee

\begin{proposition}[No blow-up in finite time] 
\label{406b} 
Fix a time $t_0$ satisfying $t_0>0$ in the expanding case and $t_0 < 0$ in the contracting case, together
with initial data $u_0$ satisfying 
\be
\label{phys1b}
0<a_0,  \qquad 0<\mut_0, \qquad  |v_0|<1.
\ee
Consider the corresponding solution of the ordinary differential system \eqref{57} defined on 
its maximal interval of existence $[t_0,T)$. 
Then, the following results hold: 
\begin{itemize}
\item In the expanding case, $T=+\infty$ and for all $T' > t_0$ there exists a constant $C_0>0$
depending on the times $t_0, T'$ and the initial data $u_0$ so that  for all $t \in [t_0,T']$ 
$$
\aligned
0 \leq H^\pm(t) &\leq \max\big( H^-(t_0), H^+(t_0) \big), 
\\
a(t) &\leq a(t_0), 
\endaligned
$$
and 
$$
\aligned
e^{-C_0 (t-t_0)} \, |w (t_0)|  
          &\leq |w(t)| 
      \leq e^{C_0 (t-t_0)} \, |w (t_0)|, 
\\
e^{-C_0 (t-t_0)} \, |z (t_0)|  
          &\leq |z(t)| 
      \leq e^{C_0 (t-t_0)} \, |z(t_0)|. 
\endaligned
$$

\item In the contracting case, $T \leq 0$ may be non zero (so that $a(t)$ may blow--up ``before'' $t$ would reach $0$), and one has 
$$
\aligned
0 \leq H^\pm(t) &\leq {t_0^2 \, a(t)^2 \over t^2 a(t_0)^2} \, \max\big( H^-(t_0), H^+(t_0) \big), 
\\
a(t_0) &\leq  a(t), 
\endaligned
$$
and $w,z$ also satisfy the bounds above for all $T' \in (t_0, T)$.  
\end{itemize}
\end{proposition} 

Of course, controling the two variables $h_1, g_1$ is sufficient to control the four variables $W_\pm, V_\pm$. Observe also that, in the expanding case, our bounds on $H^\pm, w, z$ are independent of any lower bound on $a$; such a lower bound can however be 
deduced from the equation $(\log a)_t = - (1-k^2) \, t \, \mut$ (and the control of the density $\mut$ implied by the above proposition); we do not state this here, since such a lower bound can not be ``propagated'' in the Glimm scheme, as only Riemann invariants must be used.  

Observe that, in the contracting case, the coefficient ${t_0^2 \, a(t)^2 \over t^2 a(t_0)^2} $ is greater than $1$ and a possible amplification of $H^\pm$ may take place.

\begin{proof} Consider first the expanding case. Here, we have $t>0$ and $a_t < 0$ and, since $h_1$ and $k_1$ are both positive, 
we easily find $0 \leq h_1(t) \leq h_1(t_0)$. Next, we consider the equation for $g_1$ and obtain directly 
$\big(\log|g_1| \big)' \leq 0$ so that $|g_1(t)| \leq |g(t_0)|$.  After summation, we deduce that 
$$
\aligned 
0 \leq H^\pm(t) 
& \leq {a(t) \over a(t_0)} \, \max\big( H^+(t_0), H^+(t_0) \big), 
\\
& \leq \max\big( H^+(t_0), H^+(t_0) \big), 
\endaligned
$$
since $a$ is decreasing. The Riemann invariants $w,z$ for the fluid satisfy a first--order differential system whose right--hand side is immediatly controled by $H^\pm$ and, therefore, $w,z$ can grow by a multiplicative factor $e^{C (t-t_0)}$, at most.

Consider next the case of contracting spacetimes. The function is now increasing in time and we have 
$$
\big( t^2 a^{-1} h_1 \big)' = {2t \over a} \, (h_1 - k_1 ) < 0, 
$$
since now $t<0$ and $h_1 - k_1 >0$. Similarly, one checks that $\big( \log\big( t^2 a^{-1} | g_1| \big)\big)'$ is negative and, in turn, 
we obtain 
$$
\aligned
& t^2 a(t)^{-1} h_1(t) \leq t_0^2 a(t_0)^{-1} h_1(t_0), 
\\
& t^2 a(t)^{-1} | g_1(t)| \leq t_0^2 a(t_0)^{-1} | g_1(t_0)|,
\endaligned
$$
which leads to the announced result for $H^\pm$. The argument for the fluid variables is identical to the one in the expanding case. 
\end{proof}


\section{A continuation criterion for the Einstein--Euler equations}  
\label{sec:5}
\subsection{Generalized random choice scheme}
\label{subsec:glimm}

In this section, we follow \cite{BLSS} and generalize the so-called random--choice scheme to the Einstein--Euler system  
\eqref{geometry1}--\eqref{fluid}. The scheme is constructed from an approximate solver to the {\sl generalized Riemann problem,}
i.e.~the Cauchy problem associated with piecewise constant initial data  
\be
u(t_0, \cdot) = \begin{cases}
u_l, \quad & \theta < \theta_0
\\
u_r, \quad & \theta >\theta_0,  
\end{cases} 
\ee
which we denote by $\Rcal_G(u_l,u_r; t_0, \theta_0)$. The approximate solver of interest
is constructed by evolving the solution $\uh$ of the corresponding classical 
Riemann problem $\Rcal(u_l,u_r; t_0, \theta_0)$ by the ODE system \eqref{odegeometry}--\eqref{odefluid}. 
More precisely, we have  
\be
\label{grsol}
\widetilde u(t,\theta; u_l,u_r, t_0, \theta_0) : = \uh (t,\theta; u_l,u_r, t_0, \theta_0) + \int^t_{t_0}g(\tau, S_\tau \uh (t,\theta; u_l,u_r, t_0, \theta_0))d\tau,
\ee
where $S_\tau \uh (t,\theta; t_0, \theta_0, u_l,u_r))$ is the solution of the 
ODE system \eqref{odegeometry}--\eqref{odefluid} corresponding to the initial value 
$u(t_0) = \uh (t, \theta; u_l,u_r, t_0, \theta_0)$.

To formulate the scheme of interest, we denote by $r,s>0$ the space and time mesh--lengths, respectively, 
and by $(t_k, \theta_h)$ (for $k = 0,1,\ldots$ and $h$ any integer) the mesh points of the grid, that is, 
$$
t_k := t_0 + ks,\qquad \theta_h:= hr. 
$$ 
We also set 
$$
\theta_{k,h}:= a_k + hr,
$$
where $(a_k)_{k\in\mathbb{N}}$ is a fixed equidistributed sequence in the interval $(-1,1)$. We will let $s,r\rightarrow 0$, while 
keeping the ratio $s/r$ constant. We can now define the approximate solutions $u_s=u_s(t,\theta)$ to the Cauchy problem 
for the system \eqref{geometry1}--\eqref{fluid} associated with the initial data 
$$
u_0(\theta):=u(t_0, \theta), \qquad \theta\in S^1.
$$
It is convenient to avoid the discussion of boundary conditions by assuming that $u_0$ is defined on the real line $\R$, after periodically 
extending it beyond $S^1$. 

As usual, the scheme is defined inductively. First, the initial data are approximated by a piecewise constant function    
$$
u_s(t_0, \theta) := u_0(\theta_{h+1}), \qquad \theta\in[\theta_h, \theta_{h+2}), \quad h \,  \mathrm{ even}.
$$
Then, if $u_s$ is known for all $t < t_k$, we define $u_s$ at the level $t = t_k$ as 
$$
u_s(t, \theta) = u_s(t-, \theta_{k, h+1}), \qquad \theta\in[\theta_h, \theta_{h+2}), \quad k + h \, \mathrm{ even}.  
$$  
Finally, the approximation $u_s$ is defined  in each region 
$$
t_k<t<t_{k+1}, \quad \theta_{h-1}\leq \theta< \theta_{h+1}, \quad k+h \, \mathrm{even},
$$ 
from the approximate generalized Riemann problem 
$$
\Rcal_G\big( u_s(t_k, \theta_{h-1} ),\;u_s(t_k, \theta_{h+1});\; t_k, \theta_h \big),
$$  
that is, 
$$
u_s(t, \theta) := \widetilde u\big( t, \theta;\, u_s(t_k, \theta_{h-1} ),\;u_s(t_k, \theta_{h+1});\, t_k, \theta_h \big),
$$
as introduced in \eqref{grsol}.  
  

\subsection{Sup--norm and bounded variation estimates}

We will first establish that the sup--norm of the approximate solutions remains uniformly 
bounded for all compact intervals $[t_0, T]$ with $t_0>0$ and $T < +\infty$ in the expanding case, and for 
all compact intervals $[t_0, T]$ with $t_0< T < 0$ before the mass density 
blows up in the contracting case, that is,
for all intervals $[t_0, T]$ such that 
\be
\label{eq:606} 
\sup_s \sup_{[t_0,T] \times S^1} \mut_s \leq C_1(T) \qquad \text{ in the contracting case.} 
\ee 
Subsequently, we establish a uniform bound on the total variation, again on compact time intervals. 
Recall that $s$ and $r$ must satisfy the standard stability condition  
\be
\label{eq:606b} 
\frac{s}{r}\sup_{S^1}  a_s(t,\cdot)<1,
\ee
where we have observed that the function $a_s$ provides one with a bound on all characteristic speeds associated with the classical Riemann solver. In the expanding case, $a_s$ remains bounded by the supremum of the initial data $a_0$ but, in the contracting 
case, the bound imposes that the coefficient $a_s$ does not blow up.  However, thanks to the equation \eqref{geometry3}, 
the bound \eqref{eq:606} on the mass density implies an a priori bound on the function $a_s$, so that the condition \eqref{eq:606b} can be ensured a priori from the initial data and once we know the constant $C_1(T)$. 

\begin{lemma}[Global sup--norm estimate] 
\label{uniformbounds}
Consider any bounded initial data $W_\pm^0, V_\pm^0, a_0, \mut_0, v_0$ satisfying the physical constraints 
$$
0<a_0, \qquad 0<\mut_0, \qquad |v_0|<1.
$$ 
Then, in the expanding case, for all $T > t_0$ there exists a constant $C_0$, depending on the time $T, t_0$ 
and the initial data, and a constant $C_0'$ depending on $t_0$ and the initial data so that for all $t\in[t_0, T]$ and all relevant $s$  
$$
\aligned
0<1/C_0(T) \leq a_s(t,\cdot)   &\leq \sup_{S^1} a_0(\cdot),
\\
\sup_{S^1} |W_\pm^s(t,\cdot) - 1/(2t)| + \frac{1}{t}\sup_{S^1}|V_\pm^s(t,\cdot)| & \leq C_0',
\\
\sup_{S^1} \Bigg| \log \Bigg(\frac{1+v_s(t, \cdot)}{1-v_s(t,\cdot)} \Bigg) \Bigg| + \frac{2k}{1+k^2} \, \sup_{S^1} \big| \log \mut_s(t,\cdot) \big|& \leq C_0(T). 
\endaligned
$$
On the other hand, in the contracting case the same conclusion holds for any $T \in (t_0,0)$ such that  \eqref{eq:606} holds.   
\end{lemma}

\begin{proof} We apply Propositions~\ref{homogeq} and~\ref{406b}, in which we derived ``maximum principles'' for the geometry and 
the fluid. Observe that it is at this juncture that it is important to formulate these principles in terms of the Riemann invariants associated with the geometry and the fluid variables. In each ODE step, the maximum of the geometric Riemann invariants $H_s^\pm$ 
defined as in \eqref{396} 
does not increase in the expanding case, and may only increase by a factor $1 + C \, s$ in the contracting case. 
In each Riemann problem step, the variables $H^\pm$ are non--increasing as well. The fluid variables $w,z$ 
are also similarly controled in both the ODE and the Riemann problem steps, while for the function $a$ we need to require 
that the density is a priori bounded in the contracting case; that is, the assumption  \eqref{eq:606}.  
By iterating our estimates and covering a compact time interval $[t_0, t]$,  
and after observing $( 1 + C \, s)^n 
\leq e^{-C (t-t_0)}$ provided $(t-t_0) \leq n \, s$, we obtain the desired uniform bounds. 

In the contracting case, Proposition~\ref{406b} shows that the amplification factor for $H_s^\pm$
at each time step is 
$$
{t_k^2 \, a_s(t_{k+1}, \theta)^2 \over t_{k+1}^2 a_s(t_k,\theta)^2},
$$
which we can control from the uniform bound \eqref{eq:606}, as follows, using $(\log a_s )_t = -(1- k^2) \, t \, \mut_s$ so that 
$$
\aligned
\sup_\theta {a_s(t_{k+1}, \theta)\over a_s(t_k, \theta)} 
& = \sup_\theta  \exp \Bigg( (1- k^2) \int_{t_k}^{t_{k+1}} |t''| \, \mut_s(t'', \theta) \, dt'' \Bigg) 
\\
& \leq  \exp \Bigg( (1- k^2) \, C_1(T) \, (t_{k+1} - t_k) |t_0| \Bigg) 
\endaligned
$$
for all $t_k < t_{k+1} \leq T$. By iterating this estimate over all times $\leq T$, we conclude that the sup--norm of $H_s^\pm$ 
remains uniformly bounded in the compact interval $[t_0, T]$. 
\end{proof}


We next establish the uniform bounds on the total variation of the approximate solutions. 
Denote by $u_{k, h+1}$ the value achieved by the function $u_s$ at the mesh point $(t_k, \theta_{h+1})$, i.e.
$$
u_{k,h+1} := u_s(t_k, \theta_{h+1}) \quad k+h \,\mathrm{even},
$$
and denote by $\uh _{k,h}$ the solution to the classical Riemann problem 
$$
\Rcal(u_{k-1, h}, u_{k-1,h+2}; t_{k-1}, \theta_{h+1}),
$$ 
where 
$$
u_{k,h+1} := \uh _{k,h+1} + \int_{t_{k-1}}^{t_k}g(\tau, S_\tau\uh _{k,h+1})d\tau.
$$
We divide the $(t,\theta)$--plane into diamonds $\Delta_{k,h}$ ($k+h$ even)  
with vertices $(\theta_{k-1, h}, t_{k-1})$, $(\theta_{k, h-1}, t_{k})$, $(\theta_{k, h+1},t_{k})$, 
$(\theta_{k+1, h}, t_{k+1})$. To simplify the notation, we introduce the values of $u_s$ at the vertices of $\Delta_{k,h}$ and 
the corresponding classical Riemann solvers by 
$$
u_S:= u_{k-1, h}, \qquad u_W:= u_{k, h-1}, \qquad u_E:= u_{k, h+1}, \qquad u_N := u_{k+1, h},
$$ 
$$
\uh _W:= \uh _{k, h-1}, \qquad \uh _E:= \uh _{k, h+1}, \qquad \uh _N:= \uh _{k + 1, h},
$$  
in terms of which, the strength $\eps_*( \Delta_{k,h})$ of the waves entering the diamond is defined 
as
$$
\eps_*( \Delta_{k,h}) := |\eps(\uh _W, u_S)| + |\eps(u_S, \uh _E)|,
$$
whereas the strength $\eps^*( \Delta_{k,h})$ of the waves leaving it is defined as
$$
\eps^*( \Delta_{k,h}) := |\eps(u_W, \uh _N)| + |\eps(\uh _N, u_E)|.
$$
Let $J$ be a spacelike mesh curve, that is a polygonal curve connecting the vertices $(\theta_{k, h+1}, t_k)$ of 
different diamonds, where $k+h$ is even. We say that waves $(u_{k-1, h}, \uh _{k,h+1})$ cross the curve $J$ 
if $J$ connects $(\theta_{k-1, h}, \theta_{k-1})$ to $(y_{k, h+1}, t_k)$ and similarly for $( \, \uh _{k,h-1}, u_{k-1, h})$. 
The total variation $L(J)$ of $J$ is defined as 
$$
L(J):= \sum |\eps(u_{k-1, h}, \uh _{k,h+1})| + |\eps(\uh _{k,h-1}, u_{k-1, h})|, 
$$
where the sum is taken over all the waves crossing $J$. Furthermore, we say that a curve $J_2$ is an immediate successor 
of the curve $J_1$ if they connect all the same vertices except for one and if $J_2$ lies in the future of $J_1$. For the 
difference of their total variation we have the following result. 

\begin{lemma}[Global total variation estimate] 
\label{diamond}
Let $J_1, J_2$ be two spacelike curves such that $J_2$ is an immediate successor of $J_1$ and let $\Delta_{k,h}$ be the diamond limited by these two curves. Then, for each time $T$, strictly larger than the time $t_k$ determined by $\Delta_{k,h}$, for which the assumptions of Lemma~\ref{uniformbounds} hold, there exists a constant $C_0(T)$ 
depending on $T$ and the sup bounds on the initial data, only, such that  
$$
L(J_2) - L(J_1) \leq C_0(T)\,s\eps_*(\Delta_{k,h}).
$$
\end{lemma}

\begin{proof}
By definition, we have 
$$
\aligned
L(J_2)-L(J_1) &= |\eps(u_W, \uh _N)| + |\eps(\uh _N, u_E)| - |\eps(\uh _W, u_S)| - |\eps(u_S, \uh _E)| 
\\
&=   \eps^*( \Delta_{k,h}) -  \eps_*( \Delta_{k,h}).  
\endaligned
$$ 
Observe that $|\eps(u_W, \uh _N)| + |\eps(\uh _N, u_E)| = |\eps(u_W, u_E)|$ since $\uh _N$ is just 
one of the states in the solution of the Riemann problem for $u_W, u_E$. Hence, we can write 
$$
L(J_2)-L(J_1) = X_1 + X_2,
$$
where
$$
\aligned
&X_1 : = |\eps(\uh _W, \uh _E)| - |\eps(\uh _W, u_S)| - |\eps(u_S, \uh _E)|,
\\
&X_2 : = |\eps(u_W, u_E)| - |\eps(\uh _W, \uh _E)|.
\endaligned
$$
By Proposition~\ref{inter} we have 
$$
X_1 \leq 0. 
$$
To estimate $X_2$, note that, by Lemma~\ref{uniformbounds}, $u_s$ is uniformly bounded for all $t\in[t_0, T]$, hence by Lemma~\ref{odeinter} we can write for some constant $C$ depending on $T$ and the initial data  
$$
\aligned
X_2 &\leq C \, |\eps(\uh _W, \uh _E)| \, \big( |u_W - \uh _W| + |u_E - \uh _E| \big) + C \, \big| (u_W - u_E) - (\uh _W - \uh _E) \big|
\\
    &\leq Cs \, |\eps(\uh _W, \uh _E)| \, (\sup _{t\in[t_{k-1}, t_k]}|u'_W(t)| + \sup _{t\in[t_{k-1}, t_k]}|u'_E(t)|) + C s \, |(\uh _W - \uh _E)|
\\
    &\leq Cs|\eps(\uh _W, \uh _E)| \leq Cs(|\eps(\uh _W, u_S)| + |\eps(u_S, \uh _E)|), 
\endaligned 
$$
since $|(\uh _W - \uh _E)| = O(1)|\eps(\uh _W,\uh _E)|$.
\end{proof}


\subsection{Convergence and consistency properties} 

In view of the uniform total variation bound, 
Helly's theorem allows us to extract a converging subsequence and, hence, to arrive at the following conclusion. 

\begin{theorem}[Convergence and existence theory for the Einstein--Euler equations]
\label{apptheo}
Let $u_0 = (W^0_\pm, V^0_\pm, a_0, \mut_0, v_0)$ be an initial data set for the system \eqref{geometry1}--\eqref{fluid}, which is assumed to belong to $BV(S^1)$ and satisfy the physical constraints 
$$
0<a_0,  \qquad 0<\mut_0, \qquad  |v_0|<1.
$$
Consider the approximate solutions $u_s = (W^s_\pm, V^s_\pm, a_s, \mut_s, v_s)$ constructed by the random choice scheme in Section~ \ref{subsec:glimm}. Then, the functions $u_s(t, \cdot)$ belong to $BV(S^1)$ for each fixed time $t\in[t_0, +\infty)$ in the 
expanding case and for each $t\in[t_0, T]$ in the contracting case, where $T<0$ is such that the mass--energy density 
satisfies the uniform bound
\eqref{eq:606} for some constant $C_1(T)$. 
Moreover, on any compact time interval, the solutions $u_s$ are Lipschitz continuous in time with values in $L^1$, and have uniformly bounded total variation, that is,  for all  $t, t' \in [t_0, T]$  
\be
\label{307}
TV\big( u_s(t, \cdot) \big) \leq C_0(T) \,TV(u_0),
\ee
\be
\label{308}
\int_{S^1} |u_s(t,\theta) - u_s(t', \theta)| \, d\theta \leq C_0(T) \,TV(u_0) \, \big( |t-t'| + s \big),
\ee
where $C_0$ is a constant depending on $t_0,\,T$ and the 
supremum bounds on the initial data and, additionally, on $C_1(T)$ in the contracting case.   
Furthermore, there exists a subsequence of $u_s$ converging in $L^1$, whose limit $u = (V, W, a, \mut, v)$  
is a weak solution to the system \eqref{geometry1}--\eqref{fluid} and satisfies the physical constraints  
$$
0 < a(t, \theta), \qquad 0< \mut(t, \theta), \qquad |v(t,\theta)|<1 
$$
and periodic boundary conditions.
\end{theorem}

\begin{proof}
Lemma \ref{diamond} implies that we can estimate the total variation $L(J_{k+1})$, along the curve $J_{k+1}$ connecting 
vertices of the form $(\theta_{k+1,h}, t_{k+1})$ to vertices $(\theta_{k+2,h+1}, t_{k+2})$, 
in terms of the the total variation $L(J_k)$ of the curve $J_k$. Hence, we have 
inductively
$$
\aligned 
L(J_{k+1}) &\leq L(J_k) + s\,C(T) \sum_h \eps_*(\Delta_{k,h}) 
\\
	   &\leq (1 + s\,C(T))L(J_k)  \leq e^{C(T)}L(J_1),
\endaligned 
$$
and also 
$$
L(J_1)\leq C \, TV(u_0), 
$$
which imply the claim made on $u_s$. Observe that by periodically extending $u_0$ to the real line, the proof of \eqref{308} proceeds 
along the same lines as the standard proof since, by construction, $u_s$ are periodic as 
well. By Helly's compactness theorem, there exists a subsequence, still denoted by $u_s$, converging strongly in the $L^1$ norm toward a function $u$ with bounded variation. Moreover, the uniform estimates valid for $u_s$ extend to the limit $u$.

It remains to show that $u$ is a weak solution to the system \eqref{geometry1}--\eqref{fluid}, which we express in the form 
$$
\del_t H(u,t) + \del_\theta F(u,t) + G(u,t) = 0,
$$
with obvious notation. By fixing a compactly supported smooth function $\phi$, we have  
$$
\aligned
& \int^{+\infty}_0 \int^{+\infty}_{-\infty} \big( \del_t \phi \,H(u_s,t)  + \del_\theta \phi \,F(u_s, t) + \phi \,G(u_s, t) \big) d\theta dt 
\\
& =\sum_k \sum_{h + k \, \mathrm{ odd}} \int_{k s}^{(k + 1)s}\int_{(h - 1)r }^{(h + 1)r } \big( \del_t \phi \, H(u_s, t)  + \del_\theta \phi \, F(u_s, t) + \phi \, G(u_s, t) \big) d\theta dt
\\
&
 = \Omega_s^1 + \Omega_s^2 + \Omega_s^3,
\endaligned
$$
where
$$
\aligned 
\Omega_s^1 &:= \sum_{k,h \atop h + k \, \mathrm{ odd}}  \int_{k s}^{(k + 1)s}\int_{(h - 1)r }^{(h + 1)r }\phi G(u_s, t) d\theta dt,
\endaligned
$$
$$
\aligned
\Omega_s^2 := \sum_{k,h \atop h + k \, \mathrm{ odd}}  \int_{(h - 1)r }^{(h + 1)r } 
\Bigg( 
& H\big(u_s((k + 1)s -, \theta), (k + 1)s -\big)\,\phi((k + 1)s, \theta) 
\\
& \hskip2.cm  - H\big(u_s(k s +, \theta),k s + \big) \, \phi(k s, \theta) \Bigg) \,  d\theta,
\endaligned
$$
and 
$$
\aligned
\Omega_s^3 := \sum_{k,h \atop h + k \, \mathrm{ odd}}  
 \int_{k s}^{(k + 1)s} \Bigg( 
& F \big( u_s(t, (h+1)r -), t \big) \, \phi(t, (h+1)r) dt 
\\
& \hskip1.cm  -  F \big( u_s(t, (h-1)r +), t \big) \, \phi(t, (h-1)r) \Bigg) \, dt.
\endaligned
$$
For the first term, we have  
$$
\Omega_s^1 =  O(1) \, \sum_k \sum_{h + k \, \mathrm{ odd}} \big( s^2 + r^2 \big) \, \big( s + r + |u^r_{k, h+1} - u^r_{k, h-1}| \big) \, \chi_{\mathrm{supp}\phi}, 
$$ 
where $\chi_{\mathrm{supp}\phi}$ denotes the characteristic function of the support of $\phi$; see \cite{BLSS} for further details. Since $u_s$ have uniformly bounded total variation and $\phi$ is smooth, we have $\Omega_s^1 \rightarrow 0$ for $s\rightarrow 0$.  
The second term can be rewritten as 
$$
\aligned
\Omega_s^2  
=& \sum_{k \geq 1} \sum_{h + k \, \mathrm{ odd}} \int_{(h - 1)r }^{(h + 1)r } \big( H(u_s(ks -, \theta_{k,h}), ks -) - H(u_s(ks, hr), ks) \big) \phi(k s, \theta) d\theta 
\\
&- \int^{+\infty}_{-\infty} H(u_s(t_0, \theta), t_0)\,\phi(t_0, \theta) d\theta,
\endaligned
$$
and it is well--known \cite{Glimm,Dafermos-book} that the first part on the right--hand side goes to zero as $s \rightarrow 0$ for almost every 
equidistributed sequence $(a_k)_{k\in\mathbb{N}}$. It remains to estimate $\Omega_s^3$, but after rearranging some terms, 
we obtain
$$
\aligned
\Omega_s^3 = \sum_{k} \sum_{h + k \, \mathrm{ odd}} \int_{k s}^{(k + 1)s}   \big(  
& F( u_s(t, (h+1)r +), t) 
  - F( u_s(t, (h+1)r -), t) \big) \phi(t, (h+1)r) dt = 0, 
\endaligned
$$
since $u_s$ is smooth there. Collecting the previous estimates, we conclude that, as $s \to 0$,   
$$
\aligned
& \int^{+\infty}_0 \int^{+\infty}_{-\infty} \big( \del_t \phi \, H(u_s, t)  + \del_\theta \phi F(u_s, t) + \phi G(u_s, t) \big) d\theta dt  
\\
& +   \int^{+\infty}_{-\infty} H(u_s(t_0, \theta), t_0)  \phi(t_0, \theta) d\theta 
\endaligned
$$
tends to zero. Hence, by the Lebesgue's dominated convergence theorem, passing to the limit $u_s \rightarrow u$, we arrive at 
$$
\aligned
& \int^{+\infty}_0 \int^{+\infty}_{-\infty} \big( \del_t \phi \, H(u, t)  + \del_\theta \phi F(u, t) + \phi G(u, t) \big) d\theta dt  
\\
&   +   \int^{+\infty}_{-\infty} H(u_0(t_0, \theta), t_0)  \phi(t_0, \theta) d\theta = 0.
\endaligned
$$
\end{proof}


\section{Future expanding Einstein--Euler spacetimes}
\label{sec:6} 

We are now in a position to establish our first main result.

\begin{theorem}[Expanding Einstein--Euler space\-times with Gowdy \- symmetry] 
\label{maintheo}  
Consider any BV regular, Gowdy symmetric, initial data set on $T^3$ for the Einstein--Euler equations, and assume that these initial data set has constant area $t_0>0$ and is everywhere expanding (toward the future). Then, there exists a BV regular, Gowdy symmetric spacetime $M$ with metric $g_{\alpha\beta}$ and matter fields $\mu$ and $u^\alpha$
satisfying the Einstein--Euler equations \eqref{EE1}--\eqref{eq:pressure}
in the distributional sense, and the following properties hold.
This spacetime is (up to diffeomorphisms) a Gowdy-symmetric 
future development of the initial data set and  
is globally covered by a single chart of coordinates $t$ and $(\theta,x,y) \in T^3$, with 
$$
M = [t_0, + \infty) \times T^2,   
$$
the time variable being  chosen to coincide with the area of the surfaces of symmetry.  
\end{theorem}  

This global existence result was established earlier by LeFloch and Rendall \cite{LeFlochRendall} for a different class of regularity, which provides less regularity properties on the spacetimes.  

\begin{proof}
Given an arbitrary initial data set $u = (W_\pm^0, V_\pm^0, a_0, \mut_0, v_0)$,  
we have constructed (in Section~\ref{subsec:glimm}) a sequence of approximate solutions for the ``essential'' Einstein--Euler equations \eqref{geometry1}--\eqref{fluid} and this sequence was found to converge to an exact solution $u = (W_\pm, V_\pm, a, \mut, v)$ with BV regularity to the Einstein--Euler equations understood in the sense of distributions; Cf.~Theorem~\ref{apptheo} above. 
More precisely, $u$ is locally Lipschitz continuous in time and has bounded variation with respect to the space variable, so 
for any fixed $T \in (t_0, +\infty)$ and all $t, t'\in [t_0, T]$
\be
\aligned
TV(u(t, \cdot))  & \leq C_0(T) \, TV(u_0),
\\
\int_{S^1} |u(t,\theta) - u(t', \theta)| \, d\theta   & \leq C_0(T) \, \big (|t-t'| + s \big) \,TV(u_0). 
\endaligned 
\ee
Once a solution to the first--order system \eqref{geometry1}--\eqref{fluid} has been found, we only have to recover 
the original metric coefficients $U, A$ and $\nu$ and show that the constraint equations remain satisfied for all times. However, this has 
precisely been done in Section~\ref{subsec:ftos}; hence, the proof of Theorem~\ref{maintheo} is now complete. 
\end{proof}


\section{Future contracting Einstein--Euler spacetimes}
\label{sec:7}

\subsection{Global geometry and energy functionals} 

We are now in a position to state our second main result. Observe that, for the contracting spacetimes under consideration now, a partial existence result was established in LeFloch and Rendall \cite{LeFlochRendall} for a different class of regularity, as it was already established therein that the development below covers $[t_0, t_c)$ with $t_c \leq 0$.  

\begin{theorem}[Contracting Einstein--Euler spacetimes with Gowdy symmetry] 
\label{maintheo2}  
Consider any BV regular, Gowdy symmetric, initial data set on $T^3$ for the Einstein--Euler equations, which has constant area $-t_0>0$ and is everywhere contracting toward the future. Then, there exists a BV regular, Gowdy symmetric spacetime $M$ with metric 
$g_{\alpha\beta}$ and matter fields $\mu$ and $u^\alpha$
satisfying the Einstein--Euler equations \eqref{EE1}--\eqref{eq:pressure}
in the distributional sense, and the following properties hold.
This spacetime is (up to diffeomorphisms) a Gowdy-symmetric future development of the initial data set and  
is globally covered by a single chart of coordinates $t$ and $(\theta,x,y) \in T^3$, with  
$$
M = [t_0, t_c) \times T^2  
$$
for some $t_c\leq 0$, the time variable being chosen to coincide with minus the area of the surfaces of symmetry and so that the mass energy density blows-up at $t_c$, that is, 
\be
\label{eq:606-c} 
\lim_{t \to t_c} \sup_{S^1} \mut(t, \cdot) = + \infty.
\ee 
Furthermore, when the initial data set satisfies 
\be
\label{eqn:70}
\int_{S^1} a_0^{-1} \, \Bigg(1 - \frac{4}{3} \, t_0^2 \,\frac{1 + k^2v_0^2}{1 - v_0^2} \,  \mut_0 \Bigg) \, d\theta \geq 0,
\ee
the future development exists up to the maximal time $t_c = 0$ when the area of the surfaces of symmetry shrinks to zero. 
\end{theorem}  

The theorem above seems to be essentially optimal and, for instance, we do not expect that a global--in--space, energy--type functional should allow us to decompose the space of initial data into spacetimes that blow--up (that is, $t_C<0$) and spacetimes that do not blow--up (that is, $t_C=0$). At this juncture, let us recall that the class of $T^2$--symmetric {\sl vacuum} spacetimes was studied by Isenberg and Weaver \cite{IsenbergWeaver} who could exhibit an energy that vanishes precisely for a set of exceptional solutions and also controls the behavior of the area function. Later, Smulevici \cite{Smulevici} was able to handle spacetimes satisfying the Einstein equation with a cosmological constant. 
In contrast to these problems, the Euler equations are {\sl nonlinear} partial differential equations of {\sl hyperbolic} type, so that 
the analysis of the global geometry of Einstein--Euler spacetimes is technically much more involved than the 
one of Einstein--vacuum or Einstein--Vlasov spacetimes.  
By taking advantage of the property of finite speed of propagation satisfied by the Euler equations, it is straightforward to ``localize'' our conclusions above and to construct solutions to the Einstein--Euler that are small perturbations {\sl in the $L^1$ norm} (but not in the BV semi--norm) of non--blow-up solutions (for instance homogeneous spacetimes) and, however, blow--up in a time $t_C <0$ with $t_C$ {\sl arbitrarily close} to the initial time $t_0$. 

The rest of this section is devoted to the proof of Theorem~\ref{maintheo2} and will make use of the energy functionals
$$
\aligned
E_1(t) &:= \int_{S^1} h_1 \, d\theta, \qquad \quad  E_2(t) := \int_{S^1} \big( h_1 + h_1^M \big) \, d\theta, 
\endaligned
$$
with
$$
h_1 = \frac{(tW_+ - 1/2 )^2}{2at^2} + \frac{(tW_- - 1/2 )^2}{2at^2} + \frac{V_+^2}{8at^2} + \frac{V_-^2}{8at^2},
\qquad
\quad 
h_1^M := \frac{\mut}{a}\frac{1+k^2v^2}{1-v^2}.
$$ 
Observe that $h_1$ can be obtained from the standard energy density by a simple transformation of the metric coefficient $U \mapsto 2U - \log t$. Introduce the flux 
$$
g_1:= \frac{(tW_+ -1/2)^2}{2at^2}  - \frac{(tW_- -1/2)^2}{2at^2}  + \frac{V_+^2}{8at^2} - \frac{V_-^2}{8at^2},
$$
and observe that $h_1, g_1$ satisfy the balance laws (cf.~with \eqref{eq:hc}) 
$$
\aligned
\del_t h_1 + \del_\theta (ag_1) &= \frac{a_t}{a}h_1 - \frac{2k_1}{t}, 
\\
\del_t g_1 + \del_\theta (a h_1) &= \frac{a_t}{a}g_1 - \frac{g_1}{t},
\endaligned
$$
with $k_1 := \frac{(t(W_+ + W_-) - 1)^2}{4at^2} + \frac{(V_+ - V_-)^2}{16at^2}$. 
By using the evolution equations, we obtain the following result.

\begin{lemma}
Both functionals $E_1$ and $E_2$ are monotonically increasing and, specifically, 
\be
{d E_1 \over dt} (t) = \int_{S^1}
\Bigg(
 \frac{a_t}{a}h_1 - \frac{2k_1}{t}
\Bigg)
\, d\theta \geq 0
\label{der1}
\ee 
\be
\aligned
& {d E_2 \over dt} (t) 
  = -\frac{1}{t}\int_{S^1} 
\Bigg( 
\frac{(t(W_+ + W_-) - 1)^2}{2at^2} + \frac{(V_+ + V_-)^2}{8at^2} + h_1^M + \frac{\mut}{a}\frac{3k^2 + 1}{4}
\Bigg)
\, d\theta \geq 0.
\label{der2}
\endaligned
\ee
\end{lemma}

Hence, we have 
${d E_2 \over dt} \leq -\frac{2E_2}{t}$ and thus 
$$
E_1(t) \leq E_2(t) \leq E_2(t_0)\frac{t_0^2}{t^2},
$$
which implies that the functionals $E_1$ and $E_2$ do not blow up before the singularity hypersurface $t=0$ is reached.

Observe that the arguments in Section~\ref{sec:6} can be repeated, once we establish that the metric coefficient $a$ 
and fluid density $\mu$ do not blow--up.  The new difficulty comes from the uniform bound \eqref{eq:606} 
on the mass--energy density. In fact, as discussed now, there are solution where this condition {\sl does not hold}  uniformly 
up to $t=0$.

\subsection{Revisiting the class of spatially independent solutions} 

Before we provide a rigorous proof of our non--blow--up result, we provide a simple derivation for special classes of solutions. 
By inspection of the proof of Lemma~\ref{406}, one easily checks the following slightly stronger statement.  

\begin{lemma}[No blow-up property for the homogeneous solutions in the contracting case] 
\label{405} 
Fix a time $t_0 < 0$ and initial data $u_0=(W_0^\pm, V_0^\pm, a_0, \mut_0, v_0)$ satisfying the constraints
\be
\label{phys2}
0<a_0,  \quad 0<\mut_0, \quad  |v_0|<1,
\ee
and consider the corresponding solution $u=u(t)$ of the ordinary differential system \eqref{57}, defined on 
its maximal interval of existence $[t_0,T)$ with $T \leq 0$. Then, for the geometric variables, one has  
$$
|W_\pm(t)|\leq C_0 \frac{\sqrt{a(t)}}{|t|}, \qquad |V_\pm(t)|\leq C_0 \sqrt{a(t)}, \qquad a_0 \leq a(t), 
$$
while, for the fluid variables,  
$$
\Bigg| \log \Bigg(\frac{1+v}{1-v} \Bigg) \Bigg| + \frac{2k}{1+k^2}|\log \mut| \leq C_0 \frac{a(t)}{|t|}, 
$$
where the constant $C_0$ is uniform in time and depends on the initial data $t_0, u_0$, only. 
\end{lemma}

The above lemma provides a quantitative control of all the variables (but $a$) in terms of the coefficient $a$. 
If $a(t)$ blows-up when $t$ approaches some critical time $T \in (t_0, 0)$, then at that time 
the density may tend to $0$ or $+\infty$. To be more precise and describe the nature of the blow-up for the fluid variables, we can use the energy     
$$
E_1 = \frac{(tW_+ - 1/2 )^2}{2at^2} + \frac{(tW_- - 1/2 )^2}{2at^2} + \frac{V_+^2}{8at^2} + \frac{V_-^2}{8at^2},
$$ 
introduced in the previous section, as a criterion for the blow--up of $a$. 

\begin{lemma}
\label{407} 
Given any initial data $u_0$ and the corresponding solution $u = (W^\pm, V^\pm, a, \mut, v)$ 
of the ordinary differential system \eqref{57}, one can distinguish between the following two cases: 
\begin{itemize}

\item If $E_1(t_0) \neq 0$, then the function $a$ remains bounded. 

\item If $E_1(t_0) = 0$, then the velocity remains bounded away from $\pm 1$ and the following two subcases arise: 

\begin{itemize} 

\item When $\mut_0$ is sufficiently large and, specifically, 
$$
\frac{4}{3}t_0^2 \, \mut_0 > 1, 
$$
the coefficient $a$ and consequently the density $\mut$ {\rm blow--up} for some $T \in (t_0, T)$. 

\item When $\mut_0$ is sufficiently small and, specifically, 
$$
\frac{4}{3}t_0^2 \, \frac{1 + k^2v_0^2}{1 - v_0^2} \, \mut_0 \leq 1,
$$
then $a, \mut$ remain bounded.
\end{itemize}
\end{itemize}
\end{lemma}

Our latter claim is just a special case of the final statement in Theorem \ref{maintheo2}. Observe that, in the spatially homogeneous case, where, in addition, $v = 0$, the system can be solved explicitly and the above lemma actually covers all cases. 

\begin{proof}
If $E_1(t_0) > 0$, we have 
$$
(a^{-1}E_1)' = - \frac{1}{2a^2 \, t^3} \, \Big((t(W_+ + W_-) - 1)^2 + (V_+ + V_-)^2/4\Big)\geq 0 
$$
and, therefore, 
$a^{-1}(t) \geq a^{-1}_0\, E_1(t_0) / E_1(t)$, so that the coefficient $a$ is bounded. On the other hand, if the energy initially vanishes  $E_1(t_0) = 0$, by uniqueness, it must remain identically zero for all times. Therefore, the system \eqref{57} simplifies and 
$$
t \, W = 1/2, \qquad V = 0,  \qquad a' = (1-k^2) \,  |t| \,  a\mut, 
$$
with 
$W := W_+ = W_- $, $V := V_+ = V_-$, 
and
$$
\aligned
& \Bigg(a^{-1}t\mut\frac{1+k^2v^2}{1-v^2}\Bigg)' = -a^{-1} t\mut(1-k^2)\frac{3k^2 + 1}{4t}, 
\\
&\Bigg(a^{-1}t\mut\frac{(1+k^2)v}{1-v^2}\Bigg)' = 0. 
\endaligned
$$ 
After changing the variables, we obtain  
$$
\Big( \log a^{-1}|t|\mut \Big)' = \frac{(3k^2 + 1)}{4|t|}\frac{1+v^2}{(1 - k^2v^2)}
$$
and, by integration and using $|v|<1$,  
$$
a^{-1}t\mut \leq a_0^{-1}t_0\mut_0 \Bigg(\frac{t_0}{t}\Bigg)^{(3k^2 + 1)/4} 
$$
or 
$$
a^{-1}(t) \leq a^{-1}_0 + \frac{4}{3}a_0^{-1}t_0^2\mut_0\Bigg( \Bigg(\frac{t}{t_0}\Bigg)^{3(1-k^2)/4} - 1 \Bigg). 
$$ 
Consequently, our assumption $\frac{4}{3}t_0^2\mut_0 > 1$ implies that there exists a critical time $t\in(t_0, 0)$ such that 
$1 = \frac{4}{3}t_0^2\mut_0\Big(1 - \big(\frac{t}{t_0}\big)^{3(1-k^2)/4} \Big)$, 
which leads us to $a^{-1}(t) = 0$, a contradiction. The latter claim in the lemma can be checked similarly.
\end{proof}

 
\subsection{Estimates on the geometry and fluid variables}
\label{fluid_est}  

We now consider any solution to the Einstein-Euler system defined on some interval of $[t_0, t_c)$ and we derive a bound on $\sup_{S^1} a$ which depends upon a lower bound on $t \mapsto \int_{S^1}a^{-1}(t, \theta) \, d\theta$. This property will provides us with the key ingredient of the proof of Theorem~\ref{maintheo2}.

\begin{proposition}
\label{prop:1}
If there exists a function $C=C(t)$ defined for $t \in [t_0, t_c)$ such that for ll relevant times $t$ 
\be
\label{keyassumption}
0 < C(t) \leq \int_{S^1}a^{-1}(t, \theta) \, d\theta,
\ee
then there exists also a function $C_1=C_1(t)$ such that for all such times 
\be
\label{keyproperty}
\sup_{S^1}\int_{t_0}^t\frac{\mut}{1 - v^2}(\tau, \cdot) \, d\tau \leq C_1(t).
\ee
Then, a~fortiori, $\sup_{S^1}\int_{t_0}^t \mut(\tau, \cdot) \, d\tau$ is bounded and the metric coefficient $a$ is uniformly bounded on $S^1$. 
\end{proposition}

To establish this result, we rely on the following two lemmas, which we derive first and concern the infimum of the time--average of the mass density $\int_{t_0}^t\frac{\mut}{1 - v^2}(\tau, \cdot) \, d\tau$ and its variation in space, respectively.  

\begin{lemma}[Lower bound for the averaged mass--energy density] 
\label{lem:1}
For all $t\in [t_0, t_c)$, one has 
$$
(1 + k^2) \, \Bigg(\int_{S^1}a^{-1}(t, \theta)d\theta\Bigg) \, \inf_{S^1}\int_{t_0}^t\frac{\mut}{1 - v^2}(\tau, \cdot) \, d\tau 
\leq |t_0| \, \big( E_2(t) - E_2(t_0) \big).
$$
\end{lemma}

\begin{proof} Namely, from the energy--type estimate \eqref{der2} we deduce
$$
\int_{t_0}^t \int_{S^1}a^{-1}\mut\frac{(1 + k^2)}{1 - v^2}(\tau, \theta)\,d\theta\,d\tau 
\leq |t_0| \, \big( E_2(t) - E_2(t_0) \big).
$$
Since $a^{-1}$ is monotonically decreasing in time, we obtain 
$$
\aligned
& \Bigg( \int_{S^1}a^{-1}(t, \theta)d\theta \Bigg) \,\inf_{S^1}\int_{t_0}^t\frac{\mut}{1 - v^2}(\tau, \cdot) \, d\tau  
\\
&\leq \int_{S^1}a^{-1}(t, \theta) \int_{t_0}^t \frac{\mut}{1 - v^2}(\tau, \theta)\,d\tau \, d\theta 
 \leq \int_{t_0}^t \int_{S^1} \frac{a^{-1} \mut}{1 - v^2}(\tau, \theta)\,d\theta\,d\tau. 
\endaligned
$$
\end{proof}

\begin{lemma}[Variation in space of the averaged mass--energy density] 
\label{prop:2}
For all $t\in [t_0, t_c)$ and $\theta_0, \theta_1\in S^1$, one has 
$$
\aligned
k^2\int_{t_0}^{t} \frac{\mut \, |\tau|}{1 - v^2}(\tau, \theta_1)\,d\tau  
\leq 
& (1 + k^2)\int_{t_0}^{t} \frac{\mut |\tau|}{1 - v^2}(\tau, \theta_0) \,d\tau
\\
&  + |t_0| \big(  E_1(t) - E_1(t_0) \big) 
+ 2|t| E_2(t) + 2|t_0| \, E_2(t_0). 
\endaligned
$$
\end{lemma}

\begin{proof} Setting $M:= {\mut \over 1-v^2}>0$, integrating the ``second'' Euler equation in \eqref{fluid} over some slab 
$[t_0, t]\times[\theta_0, \theta _ 1]$, and finally using $|v|<1$, we obtain (recalling that the time variable takes negative values)
$$
\aligned 
&  (1 + k^2)\int_{t_0}^{t} \tau M(\tau, \theta_1) \, d\tau 
- k^2 \int_{t_0}^{t} \tau M(\tau, \theta_0) \, d\tau 
\\
&\leq  \int_{t_0}^{t} \tau M\, (k^2 + v^2)(\tau, \theta_1) \, d\tau 
           - \int_{t_0}^{t} \tau M\, (k^2 + v^2)(\tau, \theta_0) \, d\tau
\\
& = \int_{t_0}^{t}\int_{\theta_0}^{\theta_1} \tau \frac{a_\tau}{a}g_1 \, d\theta d\tau 
          + (1 + k^2) \int_{\theta_0}^{\theta_1} \Big( t a^{-1} Mv(t, \theta)   - t_0 a^{-1} Mv(t_0, \theta) \Big) \, d\theta.
\endaligned
$$
In the above identity, the double integral term is controlled by $E_1$, i.e., 
$$
\aligned
 \int_{t_0}^{t}\int_{\theta_0}^{\theta_1} \frac{a_\tau}{a}\tau g_1 d\theta d\tau 
&\leq  - \int_{t_0}^{t}\int_{S^1} \frac{a_\tau}{a}\tau h_1 d\theta d\tau 
 \leq |t_0| \, \big( E_1(t) - E_1(t_0) \big),
\endaligned
$$
whereas, for the other two terms of the right-hand side, 
$$
\aligned
& (1 + k^2)  \int_{\theta_0}^{\theta_1} \Big( t a^{-1} Mv(t, \theta)  - t_0 a^{-1} Mv(t_0, \theta) \Big) \, d\theta
\\
& \leq t^2 \, E_2'(t) + t_0^2 \, E_2'(t_0) 
 \leq - 2t \, E_2(t) - 2t_0 \, E_2(t_0).
\endaligned
$$
Hence, we obtain  
$$
\aligned
& - k^2 \int_{t_0}^{t} \tau M(\tau, \theta_1)d\tau
\\
& \leq - (1 + k^2)\int_{t_0}^{t} \tau M(\tau, \theta_0) \, d\tau - t_0( E_1(t) - E_1(t_0)) - 2tE_2(t) - 2t_0E_2(t_0),
\endaligned
$$
which is the desired estimate. 
\end{proof}

\begin{proof}[Proof of Proposition~\ref{prop:1}]
Combining Lemmas~\ref{lem:1} and \ref{prop:2}, we obtain  
$$
\aligned 
  k^2 \sup_{S^1}\int_{t_0}^{t} \frac{\mut \, |\tau|}{1 - v^2}(\tau, \cdot) \, d\tau 
& \leq (1 + k^2)\inf_{S^1} \int_{t_0}^{t}  \frac{\mut \, |\tau|}{1 - v^2}(\tau, \cdot) \, d\tau 
\\
& \quad 
+ |t_0| \, \big( E_1(t) - E_1(t_0) \big) 
 + 2|t| \, E_2(t) + 2|t_0|\, E_2(t_0)
\\
&\leq C_0(t) \, \Bigg(\int_{S^1}a^{-1}(t, \theta)d\theta\Bigg)^{-1} + C_1(t)
\endaligned
$$ 
for some (explicit) functions $C_0(t), C_1(t)>0$. 
\end{proof}

\begin{proof}[Proof of Theorem~\ref{maintheo2}] 
We now complete the proof of our second main result. In view of Proposition~\ref{prop:1} and the earlier results in Section~4, 
it is sufficient to derive the lower bound \eqref{keyassumption} on $\int_{S^1}a^{-1}\,d\theta$. Specifically, we now check that 
if the time interval $[t_0, t_c)$ is sufficiently small or if the initial data set satisfies the condition \eqref{eqn:70}
then (for all $t\in [t_0, t_c)$ or else for all $t \in [t_0, 0)$, respectively) 
there exists a function $C(t)>0$ such that \eqref{keyassumption} holds. 
Namely, by integrating the ``first'' Euler equation in \eqref{fluid} over $[t_0, t]\times S^1$ we obtain 
$$
\aligned
 \int_{S^1} a^{-1}\mut|t|\,d\theta 
& \leq \int_{S^1} a^{-1}\mut|t| \frac{1 + k^2 v^2}{1 - v^2}\,d\theta  
\\
& \leq \int_{S^1} a_0^{-1}\mut_0|t_0| \frac{1 + k^2 v_0^2}{1 - v_0^2}\,d\theta + \int_{t_0}^t  \int_{S^1} a^{-1}\mut\frac{3k^2 + 1}{4} \, d\theta d\tau
\endaligned
$$
and, by Gronwall's inequality,
$$
 \int_{S^1} a^{-1}\mut|t|\,d\theta \leq C(t_0) \, \Bigg(\frac{t_0}{t}\Bigg)^{(3k^2 + 1)/4} 
$$
with 
$$
C(t_0) = \int_{S^1} a_0^{-1}\mut_0|t_0| (1 + k^2 v_0^2)/(1 - v_0^2) \, d\theta. 
$$
Therefore, in view of 
$(a^{-1})_t = a^{-1}\mut t (1 - k^2)$, we arrive at 
$$
\int_{S^1} a^{-1}(t, \theta)\,d\theta 
\geq \int_{S^1} a_0^{-1}(\theta)\,d\theta + \frac {4}{3} \, C(t_0)|t_0| \, \, \Bigg(\big(\frac{t}{t_0}\big)^{3(1 - k^2)/4} - 1\Bigg), 
$$
which implies \eqref{keyassumption} as long as the right-hand side above remains positive, that is, 
either on a sufficiently small interval $[t_0, t_c)$ (since the lower bound is 
always positive if $t$ is close to $t_0$), 
or else on the whole interval 
$[t_0, 0)$ if the small mass condition \eqref{eqn:70} holds.  We conclude by applying Proposition~\ref{prop:1}, which
 provides us with the required uniform bound in order to  
complete the proof of Theorem~\ref{maintheo2}. 
\end{proof}
 

\section*{Acknowledgments}  

The authors were supported by the Agence Nationale de la
Recherche via the grants ANR 2006-2--134423 (Mathematical Methods in General Relativity) and 
ANR SIMI-1-003-01 (Mathematical General Relativity. Analysis and geometry of spacetimes with low regularity).  


\end{document}